\documentclass[11pt]{article}
\usepackage[utf8]{inputenc}
\usepackage[T1]{fontenc}
\usepackage[colorlinks=true]{hyperref}
\hypersetup{urlcolor=blue, citecolor=red, linkcolor=blue}

\usepackage{amsmath,amssymb}
\usepackage[active]{srcltx}
\usepackage[all]{xy}
\input diagxy
\usepackage{xcolor}
\usepackage{enumerate}

\usepackage{mathrsfs}

\font\fr=eufm10 scaled \magstep 1 %(caracteres goticos)
\def\beq{\begin{equation}}
\def\eeq{\end{equation}}
\def\bea{\begin{eqnarray}}
\def\eea{\end{eqnarray}}
\def\beann{\begin{eqnarray*}}
\def\eeann{\end{eqnarray*}}
\def\beasn{\begin{sneqnarray}}
\def\eeasn{\end{sneqnarray}}
\def\ben{\begin{enumerate}}
\def\een{\end{enumerate}}
\def\bit{\begin{itemize}}
\def\eit{\end{itemize}}
\def\derpar#1#2{\frac{\partial{#1}}{\partial{#2}}}

\def\bm#1{\hbox{\boldmath$#1$}}

\def\moment#1#2#3{{#1}_{#2}, \ldots, {#1}_{#3}}
\def\qed{\ifvmode\Realemovelastskip\fi
{\unskip\nobreak\hfil\penalty50\hbox{}\nobreak\hfil \hbox{\vrule
height1.2ex width1.2ex}\parfillskip=0pt \finalhyphendemerits=0
\par\smallskip}}
\def\vf{\text{\fr X}}
\def\df{{\mit\Omega}}
\def\Lag{{\cal L}}
\def\L{{\cal L}}

\def\d{{\rm d}}

\def\Nat{\mathbb{N}}

\def\Real{\mathbb{R}}
\DeclareMathOperator{\rk}{rank}
\newcommand{\Cinfty}{\mathscr{C}^\infty}
\def\Tan{{\rm T}}
\def\inn{\mathop{i}\nolimits}
\newcommand*{\dd}{\mathrm{d}}

\makeatletter
\DeclareOldFontCommand{\rm}{\normalfont\rmfamily}{\mathrm}
\DeclareOldFontCommand{\sf}{\normalfont\sffamily}{\mathsf}
\DeclareOldFontCommand{\tt}{\normalfont\ttfamily}{\mathtt}
\DeclareOldFontCommand{\bf}{\normalfont\bfseries}{\mathbf}
\DeclareOldFontCommand{\it}{\normalfont\itshape}{\mathit}
\DeclareOldFontCommand{\sl}{\normalfont\slshape}{\@nomath\sl}
\DeclareOldFontCommand{\sc}{\normalfont\scshape}{\@nomath\sc}
\makeatother

\usepackage[english]{babel}
\usepackage[utf8]{inputenc}
\usepackage[T1]{fontenc}
\usepackage{eucal}
\usepackage{microtype}
\usepackage{csquotes}
\usepackage{amsthm}
\usepackage{mathtools}
\mathtoolsset{showonlyrefs}
\usepackage{amssymb,amsfonts,stmaryrd}
\usepackage{tikz}
\usetikzlibrary{cd}
\usepackage{xparse}

\usepackage{todo}

\RequirePackage{hyperref}
\theoremstyle{plain}
\newtheorem{theorem}{Theorem}[section]
\newtheorem*{theorem*}{Theorem}
\newtheorem{lemma}[theorem]{Lemma}
\newtheorem*{lemma*}{Lemma}
\newtheorem{proposition}[theorem]{Proposition}
\newtheorem*{proposition*}{Proposition}
\newtheorem{corollary}[theorem]{Corollary}
\newtheorem*{corollary*}{Corollary}
\newtheorem{definition}[theorem]{Definition}
\newtheorem*{definition*}{Definition}

\newtheorem*{example*}{Example}
\newtheorem{remark}[theorem]{Remark}
\newtheorem*{remark*}{Remark}

\newtheorem*{conjecture*}{Conjecture}

\newtheorem*{problem*}{Problem}

\newcommand*{\R}{\mathbb{R}}

\renewcommand{\L}{\mathcal{L}}
\newcommand{\bfX}{\mathbf{X}}

\newcommand{\X}{\mathfrak{X}}
\renewcommand{\d}{\mathrm{d}}
\newcommand*{\bd}{\overline{\mathrm{d}}}
\newcommand{\T}{\mathrm{T}}
\newcommand{\cT}{\mathrm{T}^\ast}
\newcommand{\Lie}{\mathscr{L}}

\newcommand{\parder}[2]{\frac{\partial #1}{\partial #2}}
\newcommand{\dparder}[2]{\dfrac{\partial #1}{\partial #2}}
\newcommand{\tparder}[2]{\partial #1/\partial #2}
\newcommand{\parderr}[3]{\frac{\partial^2 #1}{\partial #2\partial #3}}

\DeclareMathOperator{\Ima}{Im}

\title{MULTICONTACT FORMULATION \\
FOR NON-CONSERVATIVE FIELD THEORIES}
\author{\sffamily 
\sc $^a$Manuel de Le\'on
\thanks{mdeleon@icmat.es\quad ORCID: 0000-0002-8028-2348}\, ,
$^b$Jordi Gaset
\thanks{jordi.gaset@unir.net\quad ORCID: 0000-0001-8796-3149}\, ,
\sc $^c$Miguel C. Mu\~noz-Lecanda
\thanks{miguel.carlos.munoz@upc.edu\quad ORCID: 0000-0002-7037-0248} \, , \\
\sc $^b$Xavier Rivas
\thanks{xavier.rivas@unir.net\quad ORCID: 0000-0002-4175-5157}\, ,
\sc $^c$Narciso Rom\'an-Roy
\thanks{narciso.roman@upc.edu\quad ORCID: 0000-0003-3663-9861}\,.
\\[1ex]
\normalsize\itshape\sffamily 
$^a$Instituto de Ciencias Matem\'aticas,
Consejo Superior de Investigaciones Cient\'ificas\\
\normalsize\itshape\sffamily 
and Real Academia de Ciencias, Madrid, Spain.
\\[1ex]
\normalsize\itshape\sffamily 
$^b$Escuela Superior de Ingeniería y Tecnología,
Universidad Internacional de La Rioja,
Logroño, Spain.
\\[1ex]
\normalsize\itshape\sffamily 
$^c$Department of Mathematics,
Universitat Polit\`ecnica de Catalunya,
Barcelona, Spain.
}

\begin{document}

\maketitle

\begin{abstract}
A new geometric structure inspired by multisymplectic and contact geometries, which we call {\sl multicontact structure}, is developed to describe non-conservative classical field theories. Using the differential forms that define this multicontact structure as well as other geometric elements that are derived from them while assuming certain conditions, we can introduce, on the multicontact manifolds, the variational field equations  
which are stated using 
%different geometric tools, namely, 
sections, multivector fields, and Ehresmann connections on the adequate fiber bundles. Furthermore, it is shown how this multicontact framework can be adapted to the jet bundle description of classical field theories; the field equations are stated in the Lagrangian and the Hamiltonian formalisms both in the regular and the singular cases.
\end{abstract}
%%%%%%%%%%%%%%%%%%%%%%%%%%%%%%%%%%%%%%%%%%%%%%%%%%%%%%%%%%%%%%%%


\noindent\textbf{MSC\,2020 codes:}
{\sl Primary:}
70S05, 
70S10,
53D10
35R01. \\
\indent\indent\indent\indent\indent
{\sl Secondary:}
35Q99, 53C15, 53Z05, 58A10, 70G4.
%%%%%%%%%%%%%%%%%%%%%%%%%%%%%%%%%%%%%

\pagestyle{myheadings}
\markright{\small\itshape\sffamily 
{\rm M. de Le\'on} {\it et al} ---
Multicontact formulation
for non-conservative field theories.}
%%%%%%%%%%%%%%%%%%%%%%%%%%%%%%%%%%%%%%%%%%%%%%%%%%%%%%%%%%%%%%%%

\newpage

{\setcounter{tocdepth}{2}
\def\baselinestretch{1}
\small
% hack per a eliminar l'espai vertical 1em
\def\addvspace#1{\vskip 1pt}
\parskip 0pt plus 0.1mm
\tableofcontents
}

%\newpage
%%%%%%%%%%%%%%%%%%%%%%%%%%%%%%%%%%%%%%%%%%%%%%%%%%%%%%%%%%%%%%%%
\section{Introduction}

Classical field theories are a recurring theme in mathematical physics. 
The field equations are obtained in a rather intuitive way by making use of the calculus of variations, 
and that is the approach developed by T. de Donder \cite{dD-1935} 
who extended the Hamiltonian formulation for mechanics due to E. Cartan \cite{Ca-22}. 
This theory was discussed later by H. Weyl \cite{We-35} so that the theory was known as the De Donder--Weyl theory.
Later, T. Lepage \cite{Le-42} gave a more general approach.
The introduction of the notions of fiber bundles and connections by C. Ehresmann \cite{Eh1,Eh2,Eh3} 
provided the additional tools for developing the geometrical arena for a further step in the study of classical field theories.

One of the major objectives of the researchers was the possibility of extending the symplectic formalism, so successful in Hamiltonian mechanics, to the case of field theories in a covariant way. 
Let us remember that the introduction of symplectic geometry in the study of mechanics has meant a spectacular advance in the knowledge of its dynamics \cite{AM-78,Ar,LM-87}. 
It suffices to mention the study of the momentum map and the symplectic reduction theorem, 
the notions of Lagrangian submanifolds allowing a geometric interpretation of the dynamics, 
the coisotropic reduction theorem, the use of Lie groups, the introduction of geometric integrators, 
or more recently, the use of Lie groupoids, which clarify the discrete formulations of mechanics 
(both in their Hamiltonian and Lagrangian settings).

The first geometric approach is the so-called {\sl multisymplectic geometry}. 
This approximation is due to three independent groups, and is a natural extension of symplectic geometry. 
On the one hand, the seminar coordinated by W.M. Tulczyjew in Warsaw \cite{Kij1973,Kij1979},
the developments in USA by H. Goldschmidt and S. Sternberg \cite{GS-73}; and those of Pedro L. Garc\'ia-P\'erez in Spain \cite{Gc-73}.  
The multisymplectic extension is really natural: if the canonical symplectic form on the cotangent fibration 
$\Tan^*M$ of an arbitrary manifold $M$ is the differential of the canonical $1$-form (the {\sl Liouville form\/}), 
it suffices to think that the sections of the cotangent fibration are the differential $1$-forms on $M$, 
and that the vector bundle of $1$-forms $\Lambda^1\Tan^*M$ is nothing but $\Tan^*M$. 
Therefore, it suffices to extend the definition to the vector bundle of $k$-differential forms, $\Lambda^k\Tan^*M$, 
and the differential of the canonical form here is precisely a multisymplectic form.

In the case of classical field theories, one starts with a fibration 
$\pi\colon E \longrightarrow M$, where $M$ represents spacetime, and the local sections are the fields. 
In the Lagrangian formalism, one considers the fiber bundle $J^1\pi$ of local sections, where the Lagrangian density is given 
(if the field theory is of higher order, we need higher order jets bundles \cite{LR-85,book:Saunders89}).
The Hamiltonian formulation considers the $\Lambda^m\Tan^*E$ fibration of $m$-differential forms on $E$ 
(where $m=\dim M$), and two special vector fiber subbundles: $\Lambda_1^m\Tan^*E$ and $\Lambda_2^m\Tan^*E$, 
with forms that vanish when one or two arguments are vertical with respect to the $\pi$ fibration, respectively. 
The space of generalized momenta is then $J^{1*}\pi=\Lambda_2^m\Tan^*E/\Lambda_1^m\Tan^*E$. 
The Hamiltonian density determines a section of the epimorphism of vector bundles:
$h\colon J^{1*}\pi\longrightarrow \Lambda_2^m\Tan^*E$.

Since this is the natural framework for obtaining the field equations, 
the geometric model was available but the geometric structure itself had not been studied. 
M.J. Gotay proposed an abstract definition of multisymplectic form \cite{Go-91}, 
which was later studied in detail by 
several authors \cite{CCI-91,Ca96a,CIL99,IEMR-2012}, and in the paper by 
G. Martin \cite{Ma-88} a characterization of when a multisymplectic manifold is locally diffeomorphic to the model is given 
(see also \cite{LMS-2003}). 
In recent years, the research on these topics has continued in order to extend the results already known for symplectic manifolds.

Another geometric model to describe field theories is also a natural extension of symplectic geometry. 
It consists in considering the Whitney sum of $k$ copies of the cotangent bundle of a manifold $M$, namely $\oplus^k \T^*M$, 
and the family of $k$ canonical forms obtained by lifting the natural symplectic form of the cotangent bundle to the sum; 
this defines what is called  a {\sl $k$-symplectic structure}. 
Given a Hamiltonian function in that space (that is, which does not depend on the spacetime variables), this geometric framework allows us to set the Hamiltonian field equations.
The Lagrangian formalism can also be constructed on the manifold
$\oplus^k \Tan M$.
Finally, when the Lagrangian or the Hamiltonian functions depend explicitly on the spacetime variables,
the Lagrangian and the Hamiltonian formalisms are built on the manifolds
$\Real^m\times(\oplus^k \T M)$
and $\Real^m\times(\oplus^k \T^*M)$,
with $m=\dim M$, respectively, and this is the so-called
{\sl $k$-cosymplectic formulation} of first-order field theories.
(See \cite{LeSaVi2016} and references therein, for details on all these constructions). Finally, we have to mention the so-called {\sl polysymplectic formulations} of classical field theories (see, for instance, \cite{GMS-97,Ka98}).

However, there are Lagrangian and Hamiltonian densities that depend not only on the spacetime variables, 
the fields and their derivatives with respect to the spacetime coordinates (the `multivelocities', or the corresponding `multimomenta'), 
but also on certain dissipation parameters. 
These parameters are in fact a version of the action provided by the Lagrangian, 
and the corresponding equations are of non-conservative nature. 
Indeed, these field theories provide non-conservative instead of conservative quantities as in the usual case.

After agreeing on the field equations by means of a variational principle (the so-called Herglotz principle \cite{He-1930,Her-1985}), 
the aim was to identify the underlying geometrical structures. In mechanics, the framework is the {\sl contact geometry}
\cite{Banyaga2016,Geiges2008,Go-69,Kholodenko2013,LM-87} 
whose canonical model is the cotangent fibration $\Tan^*Q \times \mathbb{R}$ 
which possesses such a canonical structure, $\eta_Q = \d s - \theta_Q$, where $\theta_Q$ is the Liouville form on $\T^*Q$. 
The Hamiltonian systems on this scenario are called {\sl contact Hamiltonian systems}, and provide indeed non-conservative dynamical equations
\cite{Bravetti2017,BCT-2017,BLMP-2020,CG-2019,LGLMR-2021,LGMMR-2020,GGMRR-2019b,GG-2022,GG-2022b,Lainz2018,LIU2018,RiTo-2022}. 
The same happens in the corresponding Lagrangian picture of mechanics
\cite{CCM-2018,LGLMR-2021,LGMMR-2020,DeLeon2019,GGMRR-2019b,GG-2022b}. It is worth pointing out that contact geometry allows to study more systems than just dissipative ones \cite{LR-2022}.

So, for the case of non-conservative field theories, 
it was natural to consider a construction similar to the $k$-symplectic or the $k$-cosymplectic formulations. 
This was developed by some of the authors in \cite{GGMRR-2019,GGMRR-2020,GRR-2022,Ri-2022}.
Other less-general approaches to these kinds of geometric structures were also introduced in \cite{Bo-96,Almeida-2018,Fi2022,Mo-2008,TV-2008}.
(See also \cite{ACGL-2018,BH-2005} for other generalizations of contact geometry
related with polysymplectic geometry and other geometric structures).

Nevertheless, the analogous case to the multisymplectic one is less obvious. 
So the working strategy is twofold. On the one hand, time-dependent contact Hamiltonian systems were studied since time-dependent systems can be considered as a particular case of a field theory whose base manifold is $\mathbb{R}$, the time. 
This gives rise to a new geometric framework, the so-called {\sl cocontact structure}, constituted by a pair of 1-forms, grouping in a single structure the contact and the cosymplectic geometries
\cite{LGGMR-2022}.
On the other hand, given a fibration $\pi\colon E \longrightarrow M$,
we could begin exploring how to combine the volume form $\omega$ on $M$ (assumed always to be an oriented manifold) with
another $m$-form constructed from the canonical $m$-form defined on $\Lambda^m_2\Tan^*E$. 
Then, by studying the properties that a pair of $m$-forms should verify in order to get a model that would provide the non-conservative field equations which are obtained
from the extended Herglotz variational principle for field theories
\cite{GLMR-2022}, we would arrive to the notion of what we call {\sl multicontact structures}. 
Note that we will generalize cocontact geometry instead of contact geometry.

The procedure has been somewhat an inverse problem, since we knew the equations that had to be satisfied but not the geometry that provided them. 
% or their geometric version with that same geometry.
Thus, our aim in this paper is to define
and study the properties of this structure in the most general situation,
and then go to the particular case we are interested in,
to set the field equations, and then to develop
the Lagrangian and the Hamiltonian formalisms
for first-order non-conservative field theories,
including their variational formulation.
At the end, in the case at hand, the structure in question would be defined in the spaces $J^1\pi\times_M\Lambda^{m-1}M$
and $J^{1*}\pi\times_M\Lambda^{m-1}M$
for the Lagrangian and the Hamiltonian  cases, respectively.
It is important to point out that, although there is a previous attempt to define a generalization of the contact structure based on the multisymplectic geometry
\cite{Vi-2015},
our approach is more accurate according to the objective previously stated, which is to extend the Lagrangian and Hamiltonian field equations for non-conservative field theories.

The paper is structured as follows:
Section \ref{1} is devoted to review basic knowledge; namely,
some fundamental geometrical tools, cocontact geometry and the standard multisymplectic formulation of classical field theory.
In Section \ref{2} we introduce the {\sl multicontact structure},
we study its properties, and we restrict to the particular situation which we will be interested in later: the {\sl variational} case.
It is in this case where we state the field equations in several equivalent geometrical ways.
In Sections \ref{mlf} and \ref{mhf} we adapt the above results in order to develop
the Lagrangian and the Hamiltonian formalisms
for non-conservative field theories.
The variational formulation is discussed in Section \ref{vaf}.
Finally, Section \ref{ex} is devoted to analyse some examples; 
in particular, the time-dependent contact dynamical systems
as a particular case of our model, the vibrating string with time-dependent dissipation,
and Maxwell's equations (with charges and currents) with dissipation terms.

All the manifolds are real, second countable and of class $\Cinfty$, and the mappings are assumed to be smooth.
Sum over crossed repeated indices is understood.

%%%%%%%%%%%%%%%%%%%%%%%%%%%%%%%%%%%%%%%%%%%%%%%%%%%%%%%%%%%%%%%%
\section{Preliminary concepts}
\label{1}
%%%%%%%%%%%%%%%%%%%%%%%%%%%%%%%%%%%%%%%%%%%%%%%%%%%%%%%%%%%%%%%%

\subsection{Multivector fields and Ehresmann connections}
\label{ap:multivector}

For more information about multivector fields see, for instance, \cite{Ca96a,art:Echeverria_Munoz_Roman98,IEMR-2012}.
For the relation between multivector fields and Ehresmann connections see, for instance, \cite{art:Echeverria_Munoz_Roman98}. 

Let ${\cal M}$ be a manifold with $\dim{\cal M}=n$.
Sections of the {\sl multitangent bundle} $\Lambda^m\Tan{\cal M}\to{\cal M}$ ($m\leq n$) are called 
\textbf{$m$-multivector fields} or \textbf{multivector fields of degree $m$} on ${\cal M}$. 
They are just the contravariant skew-symmetric tensor fields of degree $m$ in ${\cal M}$
(if $m>n$ they are all zero). 
The set of $m$-multivector fields in ${\cal M}$ is denoted $\vf^m({\cal M})$.

For every $m$-multivector field $\textbf{X}\in \mathfrak{X}^m({\cal M})$ and ${\rm p}\in{\cal M}$, there exists an open neighborhood $U_{\rm p}\subset{\cal M}$ and $X_1,\dotsc,X_r \in \vf(U_{\rm p})$ such that
$$
\left. \textbf{X}\right\vert_{U_{\rm p}}= \sum_{1\leq i_1 <\dotsb< i_m \leq r} f^{i_1\dotsb i_m}X_{i_1}\wedge \dotsb \wedge X_{i_m},
$$
where $f^{i_1\dotsb i_m} \in \Cinfty(U_{\rm p})$ and $m \leq r \leq \dim{\cal M}$.
In particular, a multivector field $\textbf{X}\in \vf^m({\cal M})$ is said to be \textbf{locally decomposable} if, for every point ${\rm p}\in{\cal M}$, there exists an open neighborhood $U_{\rm p} \subset{\cal M}$ such that $$
\left.\textbf{X}\right\vert_{U_{\rm p}}=X_1\wedge\dotsb\wedge X_m \, ,\qquad \text{for some }\ X_1,\dotsc,X_m \in\vf(U_{\rm p}) \,.
$$ 
The {\sl contraction} of a multivector field $\textbf{X}\in \vf^m({\cal M})$ and a differentiable form $\Omega\in\df^k({\cal M})$
is the natural contraction between tensor fields and
is given by 
$$
i(\textbf{X})\left.\Omega \right \vert_{U_{\rm p}}= \sum_{1\leq i_1 <\dotsb< i_m \leq r} f^{i_1\dotsb i_m}\inn(X_{i_1}\wedge\dotsb \wedge X_{i_m})\,\Omega =
\sum_{1\leq i_1 <\dotsb< i_m \leq r} f^{i_1\dotsb i_m}\inn(X_{i_m})\ldots\inn(X_{i_1})\Omega \,,
$$
if $k\geq m$, and $i(\textbf{X})\left.\Omega \right \vert_{U_{\rm p}}=0$, if $k < m$.

A locally decomposable multivector field $\textbf{X}\in \vf^m({\cal M})$ 
is \textsl{locally associated}
to an $m$-dimensional distribution $\mathcal{D}$ in ${\cal M}$ 
if there exists a connected open set $U\subseteq{\cal M}$ such that $\left.\textbf{X}\right\vert_U$ is a section of $\left.\Lambda^m\mathcal{D}\right\vert_U$.
Then, locally decomposable $m$-multivector fields are locally associated to $m$-dimensional distributions. 
Multivector fields associated with the same distribution form an equivalence class $\{ \textbf{X} \}$ in $\mathfrak{X}^m({\cal M})$. 
If $\textbf{X}_1, \ \textbf{X}_2 \in \vf^m({\cal M})$ are two multivector fields locally associated with the same distribution on the open set $U \subseteq E$, then there exists 
$f \in \Cinfty(U)$ such that $\textbf{X}_1=f\,\textbf{X}_2$ (on $U$).

Given a distribution $\mathcal{D}\subseteq\Tan{\cal M}$
and ${\rm p}\in{\cal M}$, a submanifold $S\subseteq{\cal M}$
is an \textsl{integral submanifold of} $\mathcal{D}$ at ${\rm p}$ if $\Tan_{\rm p}S=\mathcal{D}_{\rm p}$.
Then, a distribution $\mathcal{D}\subseteq\Tan{\cal M}$ is an \textsl{integrable distribution} if, for every ${\rm p}\in{\cal M}$, there exists an integral submanifold of $\mathcal{D}$ containing ${\rm p}$.
Therefore, a multivector field is \textsl{integrable} if its locally associated distribution is integrable.
An $m$-dimensional submanifold $S\hookrightarrow{\cal M}$ is an \textsl{integral manifold of} $\textbf{X} \in \mathfrak{X}^m({\cal M})$ if, and only if, for every ${\rm p}\in S$, the multivector $\textbf{X}_{\rm p}$ spans $\Lambda^m\Tan_{\rm p}S$.

Now, let $\kappa\colon {\cal M}\rightarrow M$ be a fiber bundle
with $\dim{M}=m$ and $\dim{{\cal M}}=N+m$.
 A multivector field $\textbf{X} \in \vf^m({\cal M})$  is $\kappa$-\textsl{transverse} if, at every ${\rm p} \in{\cal M}$, 
we have that $\left.i(\textbf{X})(\kappa^*\beta)\right\vert_{\rm p} \neq 0$,
for every non-vanishing $\beta \in \df^m(M)$.
If $x \in M$ and ${\rm p}\in{\cal M}$, given a local section of $\kappa$, $\psi\colon U_x\subset M \rightarrow{\cal M}$, with  ${\rm p}=\psi(x)$; 
we say that $\psi$ is an \textsl{integral section} of $\textbf{X}$ at ${\rm p}$
 if $\psi(U_x)$ is the integral manifold of the multivector field $\textbf{X} \in \vf^m({\cal M})$ at ${\rm p}$.
 If a multivector field $\textbf{X} \in \mathfrak{X}^m({\cal M})$ is integrable, then it is $\kappa$-transverse if, and only if, its integral manifolds are local sections of $\kappa$.
 If $M$ is an orientable manifold
and $\omega\in\df^m(M)$
is the volume form in $M$, then
the condition that a multivector field $\textbf{X} \in \mathfrak{X}^m({\cal M})$ is $\kappa$-transverse can be written as
$\inn(\textbf{X})(\kappa^*\omega)\neq 0$; furthermore,
it is possible to take a representative $\textbf{X}$
in the class of $\kappa$-transverse multivector fields 
such that
$\inn(\textbf{X})(\kappa^*\omega)=1$.

Furthermore, the \textsl{canonical prolongation} of a section 
$\psi\colon U\subset M\to \mathcal{M}$ to 
$\Lambda^m\Tan{\cal M}$ is
the section $\psi^{(m)}\colon U\subset M\to\Lambda^m\Tan\mathcal{M}$ 
defined as $\psi^{(m)}:=\Lambda^m\Tan\psi\circ{\bf Y}_\omega$; where 
$\Lambda^m\Tan\psi\colon\Lambda^m\Tan M\to\Lambda^m\Tan\mathcal{M}$ is the natural extension of $\psi$ 
to the corresponding multitangent bundles,
and ${\bf Y}_\omega\in\vf^m(M)$ is the unique $m$-multivector field on $M$
such that $\inn({\bf Y}_\omega)\omega=1$. Then,
$\psi$ is an integral section of ${\bf X}\in\vf^m({\cal M})$ if, and only if, ${\bf X}\circ\psi=\psi^{(m)}$.

In the ambient of fiber bundles, an \textbf{Ehresmann connection} on the bundle $P\to M$
is a $\kappa$-semibasic $1$-form $\nabla$ on $P$
with values in $\Tan P$, that is, a $(1,1)$-tensor field on $P$,
such that $\inn(\nabla)\alpha =\alpha$, for every
$\kappa$-semibasic form $\alpha\in\df^1(P)$
(here $\inn(\nabla)\alpha$ denotes the usual tensorial contraction).
An Ehresmann connection splits $\Tan P$ into the {\sl vertical} and a {\sl horizontal distribution} and, in this way, $\nabla$ represents also the horizontal projector.
The connection is said to be \textsl{integrable} if its associated horizontal distribution is integrable
(the necessary and sufficient condition is that the {\sl curvature tensor} associated to $\nabla$ is zero; that is, the connection is {\sl flat}).
Then, classes of locally decomposable and
$\kappa$-transverse multivector fields $\{{\bf X}\} \subseteq \vf^m (P)$
are in one-to-one correspondence with orientable Ehresmann connection
forms $\nabla$ on $P\to M$. This correspondence is
characterized by the fact that the horizontal distribution associated
with $\nabla$ is the distribution associated with $\{ {\bf X}\}$.
In this correspondence,
classes of integrable (locally decomposable) and $\kappa$-transverse
$m$-multivector fields correspond to flat orientable Ehresmann
connections.

Finally, let ${\cal M}=J^1\pi$ be the first-order jet bundle of the bundle 
$\pi\colon E\rightarrow M$,
with natural projections $\pi^1\colon J^1\pi\rightarrow E$ and $\bar\pi^1\colon J^1\pi\rightarrow M$.
A multivector field $\textbf{X} \in \mathfrak{X}^m(J^1\pi)$ 
 is a \textbf{holonomic multivector field} (or also a \textbf{\textsc{sopde}}) if 
 (i) $\textbf{X}$ is integrable,
(ii) $\textbf{X}$ is $\bar{\pi}^1-$transverse, and
(iii) the integral sections of $\textbf{X}$ are {\sl holonomic sections} of $\bar{\pi}^1$;
that is, they are canonical lifts to $J^1\pi$ of sections of the projection $\pi$.
Note that, as a consequence, $\textbf{X}$ is locally decomposable.
If $(x^\mu,y^i,y_\mu^i)$ are coordinates in $J^1\pi$ adapted to the bundle structure, 
the general local expression for a locally decomposable  $\bar{\pi}^1-$transverse multivector field $\textbf{X}\in\vf^m(J^1\pi)$ is 
\begin{equation}
\label{Lmvf}
\textbf{X}= \bigwedge^{m}_{\mu=1} \textbf{X}_\mu = f\bigwedge^{m}_{\mu=1} \left( \frac{\partial}{\partial x^\mu} + D_\mu^i\frac{\partial}{\partial y^i} + H_{\mu\nu}^i\frac{\partial}{\partial y_\nu^i}\right)\,, \quad \text{where }f\in \Cinfty(J^1\pi) \,.
\end{equation}
This local expression
defines equivalence classes of multivector fields on $J^1\pi$.
Integral sections
$(x^\mu,y^i(x),y_\mu^i(x))$ of $\textbf{X}$ verify that $\displaystyle\derpar{y^i}{x^\mu}= D_\mu^i$ and $\displaystyle\derpar{y^i_\nu}{x^\mu}=H^i_{\mu\nu}$.
If $\textbf{X}$ is holonomic, then $D_\mu^i=y_\mu^i$ and
$\displaystyle H^i_{\mu\nu}=\frac{\partial^2y^i}{\partial x^\mu\partial x^\nu}$.
Sometimes, $\bar{\pi}^1-$transverse $m$-multivector fields which are not integrable, but whose local expressions are like \eqref{Lmvf}
with $D_\mu^i=y_\mu^i$, are called
\textbf{semi-holonomic} multivector fields.

Finally, if $\nabla$ is the Ehresmann connection in $J^1\pi$
associated with the class of $\bar\pi^1$-transverse 
and locally decomposable multivector fields represented by \eqref{Lmvf},
its local expression is
$$
\nabla= \d x^\mu\otimes\left( \frac{\partial}{\partial x^\mu} + D_\mu^i\frac{\partial}{\partial y^i} + H_{\mu\nu}^i\frac{\partial}{\partial y_\nu^i}\right) \,,
$$
and, as above, we say that $\nabla$ is \textbf{semi-holonomic} when
$D_\mu^i=y_\mu^i$, and \textbf{holonomic} or \textbf{\textsc{sopde}}
if its horizontal distribution is integrable 
and its integral submanifolds are holonomic sections.

%%%%%%%%%%%%%%%%%%%%%%%%%%%%%%%%%%%
\subsection{Notions on cocontact geometry and cocontact Hamiltonian systems}

The multicontact structure we will introduce is a generalization of {\sl cocontact geometry}.
Thus, let us provide a brief introduction to cocontact structures and cocontact Hamiltonian systems \cite{LGGMR-2022,RiTo-2022}.

\begin{definition}
Let $M$ be a manifold of dimension $2n+2$. A \textbf{cocontact structure} on $M$ is a pair $(\tau,\eta)$ of 1-forms on $M$ such that $\tau$ is closed and such that $\tau\wedge\eta\wedge(\d\eta)^{\wedge n}$ is a volume form on $M$. Under these hypotheses, $(M,\tau,\eta)$ is called a \textbf{cocontact manifold}.
\end{definition}

In a cocontact manifold $(M,\tau,\eta)$ there are two vector fields $R_t,R_s\in\X(M)$ which satisfy the conditions
	$$ \begin{dcases}
		i(R_t)\tau = 1\,,\\
		i(R_t)\eta = 0\,,\\
		i(R_t)\d\eta = 0\,,
	\end{dcases}\qquad\begin{dcases}
		i(R_s)\tau = 0\,,\\
		i(R_s)\eta = 1\,,\\
		i(R_s)\d\eta = 0\,.
	\end{dcases} $$
The vector fields $R_t$ and $R_s$ are called the \textsl{time Reeb vector field} and the \textsl{contact Reeb vector field}, respectively.

If $(M,\tau,\eta)$ is a cocontact manifold and $H\in\Cinfty(M)$,
then the tuple $(M,\tau,\eta,H)$ is a \textbf{cocontact Hamiltonian system}. The equations 
\beq
\inn(X_H)\tau=1  \,, \qquad \inn(X_H)\eta=-H  \,, \qquad
\inn(X_H)\d\eta = \d H - (\Lie_{R_s}H)\eta - (\Lie_{R_t}H)\tau
\label{hamilton-cocontact-eqs}
\eeq
are the \textsl{cocontact Hamiltonian equations for vector fields}
and a solution $X_H$ to them is said to be a
\textsl{cocontact Hamiltonian vector field} associated with $H$. The vector field $X_H\in\vf (M)$ is unique.

If ${\bf c}\colon I\subset\R\to M$ is a curve and ${\bf c}'\colon I\subset\R\to \Tan M$ is the canonical lift of 
${\bf c}$ to the tangent bundle $\Tan M$, then the equations
$$
\inn({\bf c}')\tau=1
\,, \qquad \inn({\bf c}')\eta = -H\circ{\bf c}
\,, \qquad
\inn({\bf c}')\d\eta=\left(\d H-(\Lie_{R_s}H)\eta-(\Lie_{R_t}H)\tau\right)\circ{\bf c}
$$
are the \textsl{cocontact Hamiltonian equations for curves}
and their solutions are the integral curves of the contact Hamiltonian vector fields solution to 
\eqref{hamilton-cocontact-eqs}.

In Section \ref{cocontactsection}, it is shown how the cocontact formulation is recovered from the multicontact setting introduced in
Section \ref{2}.

%%%%%%%%%%%%%%%%%%%%%%%%%%%%%%%%%%%
\subsection{Review on multisymplectic field theory}
\label{nmft}

For more details on the Lagrangian formalism of classical field theories see, for instance, 
\cite{Ald1980,LMS-2004,EMR-96,art:Echeverria_Munoz_Roman98,Gc-73,GS-73,
art:Roman09,book:Saunders89},
and for the Hamiltonian formalism 
\cite{CCI-91,LMM-96,EMR-99b,
EMR-00b,GMM-2022,HK-04,Kru2002,MS-98,Pau2002,art:Roman09}.

A first-order classical field theory
is described by the following elements:
First, let $\pi\colon E\to M$ be the configuration bundle,
with $\dim M=m$ and $\dim E=n+m$,
where $M$ is an orientable manifold with volume form
$\omega_{_M}\in\df^m(M)$, and the first-order jet bundle
$\pi^1\colon J^1\pi\to E$, which is also a bundle over $M$
with projection
$\bar\pi^1=\pi\circ\pi^1\colon J^1\pi\longrightarrow M$,
and $\dim\,J^1\pi=nm+n+m$.
We denote by $(x^\mu,y^i,y^i_\mu)$
($\mu = 1,\ldots,m$; $i=1,\ldots,n$)
natural coordinates in $J^1\pi$ adapted to the bundle structure
and such that
$\omega:=\bar\pi_1^*\omega_{_M}=\d x^1\wedge\cdots\wedge\d x^m\equiv\d^mx$.
Second, 
the {\sl Lagrangian density} is a $\bar\pi^1$-semibasic $m$-form on
$J^1\pi$ that can be expressed as $\Lag =L\omega$,
where $L\in\Cinfty(J^1\pi)$
is the {\sl Lagrangian function} associated with $\Lag$ and $\omega$.
The bundle $J^1\pi$ is endowed with a canonical structure, the {\sl canonical endomorphism}
${\rm J}$, which is a $(1,2)$-tensor field whose local expression is
 $\displaystyle {\rm J}=\left(\d y^i-y^i_\mu\d x^\mu\right)\otimes
\derpar{}{y^i_\nu}\otimes\derpar{}{x^\nu}$. Then
the {\bf Poincar\'e--Cartan $m$} and {\bf $(m+1)$-forms}
associated with~$\Lag$ are defined as
$\Theta_{\Lag}:=\inn({\rm J})\d\Lag+\Lag\in\df^{m}(J^1\pi)$ and
$\Omega_{\Lag}:= -\d\Theta_{\Lag}\in\df^{m+1}(J^1\pi)$,
and they have the following local expressions:
\begin{align}
\Theta_{\mathcal{L}} &= \frac{\partial L}{\partial y^i_\mu}\d y^i\wedge\d^{m-1}x_\mu -\left(\frac{\partial L}{\partial y^i_\mu}y^i_\mu-L\right)\d^m x \,,
\label{ThetaL}
\\
\Omega_{\Lag} &=
-\frac{\partial^2L}{\partial y^j_\nu\partial y^i_\mu}
\,\d y^j_\nu\wedge\d y^i\wedge\d^{m-1}x_\mu
-\frac{\partial^2L}{\partial y^j\partial y^i_\mu}\,\d y^j\wedge
\d y^i\wedge\d^{m-1}x_\mu
\nonumber  \\ & \quad
+\frac{\partial^2L}{\partial y^j_\nu\partial y^i_\mu}\,y^i_\mu\,
\d y^j_\nu\wedge\d^mx  +
\left(\frac{\partial^2L}{\partial y^j\partial y^i_\mu}y^i_\mu
 -\derpar{L}{y^j}+\frac{\partial^2L}{\partial x^\mu\partial y^j_\mu}
\right)\d y^j\wedge\d^mx \,,
\nonumber 
\end{align}
where $\displaystyle\d^{m-1}x_\mu\equiv\inn\left(\derpar{}{x^\mu}\right)\d^mx$.
The pair $(J^1\pi,\Omega_\Lag)$ is called a {\sl Lagrangian system}.
The Lagrangian function and the corresponding Lagrangian system are {\sl regular} if $\Omega_{\Lag}$ is a {\sl multisymplectic
form} (i.e., $1$-non\-degenerate); elsewhere they are
 {\sl singular} or {\sl non-regular}.
This regularity condition is locally equivalent to demand that the Hessian matrix
$\displaystyle\left(\frac{\partial^2L}
{\partial y^i_\mu\partial y^j_\nu}\right)$
is regular everywhere.

The Lagrangian field equations can be derived from the so-called {\sl Hamilton variational principle} which states that,
if $\Gamma(\pi)$ denotes the set of sections of $\pi$,
the {\sl variational problem} for a Lagrangian system $(J^1\pi,\Omega_\Lag)$
is the search of the critical (or
stationary) sections of the functional 
\begin{align*}
 {\bf L} \colon \ \Gamma(\pi) &\longrightarrow \Real
 \\ 
 \phi &\longmapsto \int_M(j^1\phi)^*\Theta_\Lag  \,,
\end{align*}
 with respect to the variations of $\phi$ given
 by $\phi_t =\sigma_t\circ\phi$, where $\{\sigma_t\}$ is a
 local $1$-parameter group of any compact-supported $\pi$-vertical vector field
 $Z\in\vf(E)$; that is,
$\displaystyle \frac{\d}{\d t}\Big\vert_{t=0}\int_M\big(j^1\phi_t\big)^*\Theta_\Lag = 0$.
Then, for a Lagrangian system  $(J^1\pi,\Omega_\Lag)$,
the Lagrangian field equations derived from this variational principle can be stated geometrically in several alternative ways.
In particular,
for a section $\phi\in\Gamma(\pi)$,
these equations are equivalently stated as:
 \begin{enumerate}[{\rm (1)}]\itemsep=0pt
 \item
 $(j^1\phi)^*\inn (X)\Omega_\Lag= 0$, 
 for every $X\in\vf (J^1\pi)$.
  \item
$\inn((j^1\phi)^m)(\Omega_\Lag\circ j^1\phi)=0$.
\item
 $j^1\phi$ is an integral section of a class of holonomic multivector fields
 $\{ {\bf X}_\Lag\}\subset\vf^m(J^1\pi)$ satisfying
$$
 \inn ({\bf X}_\Lag)\Omega_\Lag=0 \, , \quad
 \text{\rm for every ${\bf X}_\Lag\in\{ \bfX_\Lag\}$}   \,.
$$
\item
 $j^1\phi$ is an integral section of a holonomic connection
$\nabla_\Lag$ in $J^1\pi$ satisfying 
$$
\inn(\nabla_\Lag)\Omega_\Lag=(m-1)\Omega_\Lag  \,.
$$
 \end{enumerate}
In a natural system of coordinates in $J^1\pi$, if $\phi=(x^\mu ,y^i(x^\nu))$,
then
$\displaystyle j^1\phi=\Big(x^\mu ,y^i(x^\nu),\derpar{y^i}{x^\mu}(x^\nu)\Big)$
and the above equations lead to
the Euler--Lagrange equations
$$
 \derpar{L}{y^i}\circ j^1\phi-
\derpar{}{x^\mu}\left(\derpar{L}{y_\mu^i}\circ j^1\phi\right)= 0 \,.
$$

For the Hamiltonian formalism of classical field theories,
consider the {\sl extended multimomentum bundle} 
${\cal M}\pi\equiv\Lambda_2^m\Tan^*E$, that is the bundle of $m$-forms on
$E$ vanishing by contraction with two $\pi$-vertical vector fields
(or equivalently, the set of affine maps from
$J^1\pi$ to $\pi^*\Lambda^m\Tan^*M\simeq\Real$),
which is endowed with natural coordinates $(x^\nu,y^i,p^\nu_i,p)$
adapted to the bundle $\pi\colon E\to M$,   and such that
$\omega=\d^mx$; so
$\dim\, {\cal M}\pi=nm+n+m+1$.
We define the {\sl restricted multimomentum bundle}
$J^{1*}\pi\equiv{\cal M}\pi/\Lambda^m_1\Tan^*E$
(where $\Lambda^m_1\Tan^*E$ is the bundle of $\pi$-semibasic $m$-forms on
$E$);
whose natural coordinates are $(x^\mu,y^i,p_i^\mu)$, and so 
$\dim\, J^{1*}\pi=nm+n+m$.
We have the natural projections
$$
\bar\kappa\colon J^{1*}\pi\to M \,, \quad
\kappa\colon J^{1*}\pi\to E \,, \quad
\mathfrak{p}\colon {\cal M}\pi\to J^{1*}\pi \,.
$$

Associated with a Lagrangian function $L$ there are the {\sl extended Legendre map},
$\widetilde{FL}\colon J^1\pi\to {\cal M}\pi$,
and the {\sl restricted Legendre map}, $FL:=\mathfrak{p}\circ\widetilde{FL}\colon J^1\pi\to J^{1*}\pi$, 
which are locally given by the expressions
 \begin{equation}
    \begin{aligned}
 &\widetilde{FL}^{\,*}x^\nu = x^\nu\,, \qquad
 \widetilde{FL}^{\,*}y^i = y^i\,,\qquad
 \widetilde{FL}^{\,*}p_i^\nu=\displaystyle\derpar{L}{y^i_\nu}
 \label{FL1}
 \,,\qquad
 \widetilde{FL}^{\,*}p =\displaystyle L-y^i_\nu\derpar{L}{y^i_\nu} \,,
 \\
 &{FL}^*x^\nu = x^\nu \,,\qquad
 {FL}^*y^i = y^i \,,\qquad
 {FL}^*p_i^\nu =\displaystyle\derpar{L}{y^i_\nu} \,.
%  \label{FL2}
\end{aligned}
 \end{equation}
 The Lagrangian $L$ is regular if, and only if,
 $FL$ is a local diffeomorphism,
and $L$ is said to be {\sl hyperregular} when $FL$ is a global diffeomorphism.

As ${\cal M}\pi$ is a subbundle of $\Lambda^m\Tan^*E$,
it is endowed with a canonical form $\widetilde\Theta\in\df^m({\cal M}\pi)$
(the ``tautological form'') which is defined as follows:
if $(y,\xi)\in\Lambda_2^m\Tan^*E $, with $y\in E$, $x=\pi(y)$, and
$\xi\in\Lambda_2^m\Tan_x^*E$; then,
for every $X_1,\ldots,X_m\in\Tan_{(y,\xi)}({\cal M}\pi)$,
\[
\widetilde\Theta ((y,\xi);\moment{X}{1}{m}):=
\xi(y;\Tan_{(y,\xi)}(\kappa\circ\mathfrak{p})(X_1),\dotsc ,\Tan_{(y,\xi)}(\kappa\circ\mathfrak{p})(X_m)) \,.
\]
Then we also have the multisymplectic form
 $\widetilde\Omega:=-\d\widetilde\Theta\in\df^{m+1}({\cal M}\pi)$.
They are known as the {\sl multimomentum Liouville $m$ and $(m+1)$-forms},
whose local expressions are
\beq
\widetilde\Theta = p_i^\mu\d y^i\wedge\d^{m-1}x_\mu+p\,\d^mx
\, , \qquad
\widetilde\Omega = -\d p_i^\mu\wedge\d y^i\wedge\d^{m-1}x_\mu-\d p\wedge\d^mx \,.
\label{Liform}
\eeq
  
Now, when $L$ is a regular or hyperregular Lagrangian,
we can construct the map  
${\bf h}:=\widetilde{FL}\circ FL^{-1}$
(which is a section of the projection $\mathfrak{p}\colon{\cal M}\pi\to J^{1*}\pi$,
and is called a {\sl Hamiltonian section\/}).
\beq
 \begin{array}{ccccc}
\begin{picture}(15,52)(0,0)
\put(0,0){\text{$J^1\pi$}}
\end{picture}
&
\begin{picture}(65,52)(0,0)
 \put(20,27){\text{$\widetilde{FL}$}}
 \put(25,6){\text{$FL$}}
 \put(0,7){\vector(2,1){65}}
 \put(0,4){\vector(1,0){65}}
\end{picture}
&
\begin{picture}(15,52)(0,0)
 \put(0,0){\text{$J^{1*}\pi$}}
 \put(0,41){\text{${\cal M}\pi$}}
 \put(5,38){\vector(0,-1){25}}
 \put(-5,22){\text{$\mathfrak{p}$}}
 \put(10,13){\vector(0,1){25}}
 \put(15,22){\text{${\bf h}$}}
\end{picture}
\end{array} 
\label{1stdiag}
\eeq
Then,
 associated with ${\bf h}$, we can define the so-called {\bf Hamilton--Cartan $m$-form} of $J^{1*}\pi$ as
$\Theta_{\bf h}:={\bf h}^*\widetilde\Theta\in\df^m(J^{1*}\pi)$,
and the {\bf Hamilton--Cartan $(m+1)$-form}
$\Omega_{\bf h}:=-\d\Theta_{\bf h}=
{\bf h}^*\widetilde\Omega\in\df^{m+1}(J^{1*}\pi)$ which is a multisymplectic form in $J^{1*}\pi$.
Since, locally, the Hamiltonian section is specified by
${\bf h}(x^\mu,y^i,p^\mu_i)=(x^\mu,y^i,p^\mu_i,p=-H(x^\nu,y^i,p_j^\nu))$, the local expressions of these differential forms are
\beq
 \Theta_{\bf h}= p_i^\nu\d y^i\wedge\d^{m-1}x_\mu -H\,\d^mx 
\,, \qquad
\Omega_{\bf h}= -\d p_i^\mu\wedge\d y^i\wedge\d^{m-1}x_\mu +
\d H\wedge\d^mx \,.
\label{HCform}
\eeq
Note that
 $\widetilde{FL}^*\widetilde\Theta={FL}^*\Theta_{\bf h}=\Theta_{\Lag}$
 and $\widetilde{FL}^*\widetilde\Omega={FL}^*\Omega_{\bf h}=\Omega_{\Lag}$.
 The pair $(J^{1*}\pi,\Omega_{\bf h})$
 is called the {\sl Hamiltonian system} associated with the (hyper)regular Lagrangian system $(J^1\pi,\Omega_\Lag)$.

The above construction can be done independently of having a starting Lagrangian system, 
simply giving any Hamiltonian section of the projection $\mathfrak{p}\colon{\cal M}\pi\to J^{1*}\pi$. 
Then, as in the Lagrangian formalism, in order to obtain the Hamiltonian field equations
we can state the so-called {\sl Hamilton--Jacobi variational principle}
as follows: 
the variational problem for the Hamiltonian system $(J^{1*}\pi,\Omega_{\bf h})$
is the search of the critical sections of the functional 
\begin{align*}
 {\bf H} \colon \ \Gamma(\bar\kappa) &\longrightarrow \Real
 \\ 
 \psi &\longmapsto \int_M\psi^*\Theta_{\bf h} \,,
\end{align*}
 with respect to the variations of $\psi$ given
 by $\psi_t =\sigma_t\circ\psi$, where $\{\sigma_t\}$ is a
 local one-parameter group of any compact-supported $\bar\kappa$-vertical vector field
 $Z\in\vf(J^{1*}\pi)$; that is,
$\displaystyle \frac{\d}{\d t}\Big\vert_{t=0}\int_M\psi_t^*\Theta_{\bf h}=0$.
Then, the Hamiltonian field equations derived from this variational principle,
which are called {\sl Hamilton--de Donder--Weyl equations},
can be set geometrically in several alternative ways.
In particular,
for a section $\psi\in\Gamma(\bar\kappa)$,
these equations are equivalently stated as:
 \begin{enumerate}[{\rm (1)}]\itemsep=0pt
 \item
 $\psi^*\inn (X)\Omega_{\bf h}= 0$, 
 for every $X\in\vf (J^{1*}\pi)$.
  \item
$\inn(\psi^{(m)})(\Omega_{\bf h}\circ\psi)=0$.
\item
 $\psi$ is an integral section of a class of integrable and
 $\bar\kappa$-transverse multivector fields
 $\{ {\bf X}_{\bf h}\}\subset\vf^m(J^{1*}\pi)$ satisfying that
$$
 \inn ({\bf X}_{\bf h})\Omega_{\bf h}=0 \, , \qquad
 \text{\rm for every ${\bf X}_{\bf h}\in\{ {\bf X}_{\bf h}$\}} \,.
$$
\item
$\psi$ is an integral section of an integrable connection
$\nabla_{\bf h}$ in $J^{1*}\pi$ satisfying the equation
$$
\inn(\nabla_{\bf h})\Omega_{\bf h}=(m-1)\Omega_{\bf h}\,.
$$
 \end{enumerate}
In a natural system of coordinates in $J^{1*}\pi$, if $\psi=(x^\mu ,y^i(x^\nu),p^\mu_i(x^\nu))$,
the above equations lead to
$$
 \derpar{(y^i\circ\psi)}{x^\mu}=
 \derpar{H}{p^\mu_i}\circ\psi
\, ,\qquad
 \derpar{(p_i^\mu\circ\psi)}{x^\mu}=
 - \derpar{H}{y^i}\circ\psi \,.
$$

If $(J^{1*}\pi,\Omega_{\bf h})$ is the Hamiltonian system associated with a
(hyper)regular Lagrangian system $(J^1\pi,\Omega_\Lag)$ then,
since the restricted Legendre map $FL$ is a diffeomorphism, it is evident that the solutions 
to the Lagrangian field equations for the system $(J^1\pi,\Omega_\Lag)$
are in one-to-one correspondence with the solutions 
to the Hamilton--de Donder--Weyl equations for the Hamiltonian system $(J^{1*}\pi,\Omega_{\bf h})$.

In the case of singular Lagrangians, in order to ensure the existence of a Hamiltonian description 
we need to impose some minimal regularity conditions.
In particular, we consider the case of 
{\sl almost-regular Lagrangians} $L\in\Cinfty(J^1\pi)$, which are those such that:
\begin{enumerate}[{\rm (1)}]\itemsep=0pt
\item
 $P_0:=FL (J^1\pi)$ is a closed submanifold of $J^{1*}\pi$,
\item
 $FL$ is a submersion onto its image,
\item
for every $\bar y\in J^1\pi$, the fibers $FL^{-1}(FL(\bar y))$, 
 are connected submanifolds of $J^1\pi$.
\end{enumerate}
Therefore
we have that the submanifolds $P_0:= FL(J^1\pi)$
and $\widetilde P_0:=\widetilde{FL}(J^1\pi)$
are fiber bundles over $E$ and $M$,
and $\widetilde P_0$ is transverse to the projection
$\mathfrak{p}\colon{\cal M}\pi\to J^{1*}\pi$ and diffeomorphic to $P_0$.
We denote by
$\jmath_0\colon P_0\hookrightarrow J^{1*}\pi$ and
$\widetilde\jmath_0\colon\widetilde P_0\hookrightarrow{\cal M}\pi$
the corresponding embeddings.
 This diffeomorphism is denoted by
 $\widetilde{\mathfrak{p}}\colon\widetilde P_0\to P_0$,
 and is just the restriction of $\mathfrak{p}$ 
 to $\widetilde P_0$. Thus, we have the diagram
 \beq
 \begin{array}{cccc}
\begin{picture}(20,52)(0,0)
\put(0,0){\text{$J^1\pi$}}
\end{picture}
&
\begin{picture}(65,52)(0,0)
 \put(7,28){\text{$\widetilde{FL}_0$}}
 \put(24,7){\text{$FL_0$}}
 \put(0,7){\vector(2,1){65}}
 \put(0,4){\vector(1,0){65}}
\end{picture}
&
\begin{picture}(90,52)(0,0)
 \put(5,0){\text{$P_0$}}
 \put(5,42){\text{$\widetilde P_0$}}
 \put(5,13){\vector(0,1){25}}
 \put(10,38){\vector(0,-1){25}}
 \put(-10,22){\text{$\widetilde{\bf h}$}}
 \put(12,22){\text{$\widetilde{\mathfrak{p}}$}}
 \put(30,45){\vector(1,0){55}}
 \put(30,4){\vector(1,0){55}}
 \put(48,12){\text{$\jmath_0$}}
 \put(48,33){\text{$\widetilde\jmath_0$}}
 \end{picture}
&
\begin{picture}(15,52)(0,0)
 \put(0,0){\text{$J^{1*}\pi$}}
 \put(0,41){\text{${\cal M}\pi$}}
 \put(10,38){\vector(0,-1){25}}
 \put(0,22){\text{${\mathfrak{p}}$}}
\end{picture}
\\
& &
\begin{picture}(90,35)(0,0)
 \put(10,35){\vector(1,-1){35}}
 \put(5,11){\text{$\bar\kappa_0$}}
 \put(90,11){\text{$\bar\kappa$}}
 \put(100,35){\vector(-1,-1){35}}
\end{picture}
 &
\\
& & \qquad M &
 \end{array}
 \label{lastdiag}
 \eeq
where $FL_0$ is the restriction map of
$FL$ onto $P_0$, defined by
$FL=\jmath_0\circ FL_0$,
and the same for $\widetilde{FL}_0$.
Then, taking $\widetilde{\bf h}:=\widetilde{\mathfrak{p}}^{-1}$,
we define the Hamilton--Cartan forms
\[
\Theta^0_{\bf h}=(\widetilde\jmath_0\circ\widetilde{\bf h})^*\widetilde\Theta\in\df^m(P_0) \,,\qquad
\Omega^0_{\bf h}=-\d\Theta^0_{\bf h}=
(\widetilde\jmath_0\circ\widetilde{\bf h})^*\widetilde\Omega\in\df^{m+1}(P_0) \,,
 \]
which verify that $\Theta_\mathcal{L}=FL_0^{\ *}\,\Theta^0_{\bf h}$ and
$\Omega_\mathcal{L}=FL_0^{\ *}\,\Omega^0_{\bf h}$.
Furthermore, there exists a Hamiltonian
function $H_0\in\Cinfty(P_0)$
such that
$E_\mathcal{L}=FL_0^{\ *}\,H_0$,
and therefore the coordinate expression of these forms are
\beq
\Theta^0_{\bf h}=\jmath_0^{\,*}(p_i^\mu\d y^i\wedge\d^{m-1}x_\mu)-H_0\,\d^mx
\, , \qquad
\Omega^0_{\bf h}=\jmath_0^{\,*}(-\d p_i^\mu\wedge\d y^i\wedge\d^{m-1}x_\mu)+\d H_0\wedge\d^mx \,.
\label{Thetacero}
\eeq
In general, $\Omega_{\bf h}^0$ is a premultisymplectic form.
The pair $(P_0,\Omega^0_{\bf h})$ is the {\sl Hamiltonian system}
associated with the almost-regular Lagrangian system $(J^1\pi,\Omega_\Lag)$.

In this framework,
the Hamilton--de Donder--Weyl equations
for the Hamiltonian system $(P_0,\Omega^0_{\bf h})$
are stated like in the regular case.

%%%%%%%%%%%%%%%%%%%%%%%%%%%%%%%%%%%%%%%%%%%%%%%%%%%%%%%%%%%%%%%%
\section{Multicontact geometry}
\label{2}
%%%%%%%%%%%%%%%%%%%%%%%%%%%%%%%%%%%%%%%%%%%%%%%%%%%%%%%%%%%%%%%%

\subsection{Motivational guidelines}

In this section, we present a new geometric structure
which extends the notion of {\sl contact structure}
used in contact mechanics,
in order to describe non-conservative classical field theories. 
This new definition must be a generalization of the so-called {\sl co-contact structure}
first introduced in \cite{LGGMR-2022} to describe time-dependent dissipative mechanical systems. 
It should also be consistent with the notions of {\sl $k$-contact} and {\sl $k$-cocontact manifolds} for non-conservative field theories \cite{GGMRR-2019,GGMRR-2020,GRR-2022,Ri-2022}
in the sense that these models must lead to the same equations when applied to the same situations. 
This new geometric structure should lead to a system of partial differential equations describing the field theory. 
In the Lagrangian case, these equations should be compatible with the ones derived in \cite{GLMR-2022,Geor-2003,LPAF2018} from variational principles. These equations are the {\sl Herglotz--Euler--Lagrange equations} for field theories
\cite{GLMR-2022}, which read
$$
\frac{\partial}{\partial x^\mu}
\left(\frac{\displaystyle\partial L}{\partial
y^i_\mu}\right)=
\frac{\partial L}{\partial y^i}+
\displaystyle\frac{\partial L}{\partial s^\mu}\displaystyle\frac{\partial L}{\partial y^i_\mu} \,.
$$

In order to find this new structure, we first consider the usual geometry to describe first-order Lagrangian field theories: the fiber bundle $J^1\pi$ of $\pi:E\rightarrow M$. We also consider $\Lambda^{m-1}(\Tan^*M)$ which, based on the Herglotz's variational principle for fields, is the natural structure to define the new variables $s^\mu$ that represent the dependence of the Lagrangian on the action. In these fiber bundles we have several natural forms: the Poincar\'e--Cartan $m$-form associated with a Lagrangian function $L$ in $J^1\pi$, the tautological form associated to $\Lambda^{m-1}(\cT M)$, and a volume form on $M$ (which is relevant in the cocontact formalism). 
Based on these objects, we want to obtain a new form defined in an appropriate extension of the jet bundle,
whose coordinate expression reads
\begin{equation}\label{thetacoor1}
    \Theta_{\mathcal{L}} =-\frac{\partial L}{\partial y^i_\mu}\d y^i\wedge\d^{m-1}x_\mu +\left(\frac{\partial L}{\partial y^i_\mu}y^i_\mu-L\right)\d^m x+\d s^\mu\wedge \d^{m-1}x_\mu\,,
\end{equation}
for a Lagrangian function $L$ defined in that jet bundle extension. 
The new variables $s^\mu$ must give account for the ``non-conservation''
and they are the ``multicontact variables''.
This form will be used to characterize the field equations so that we
reach the {\sl Herglotz--Euler--Lagrange equations}.

Then, we proceed to generalize the structure as much as possible,
as long as it allows us to derive the field equations. 
This is important because some of the applications naturally require manifolds which do not have the canonical structure. This is the case for multisymplectic systems and it is expected that it will also be the case for multicontact systems.

It is important to note that the structure we will define is not the same as that introduced in \cite{Vi-2015}
which is too general for the objectives previously mentioned.

{\it Notation\/}:
In this section we work with distributions.
So, the vector fields associated with a distribution ${\cal D}$ will be denoted
$\Gamma({\cal D})$, the sections of ${\cal D}$.
In particular, for a differential form $\omega$, $\ker\omega$ is a distribution and $\Gamma(\ker\omega)$ is the set of the associated vector fields.

%%%%%%%%%%%%%%%%%%%%%%%%%%%%%%%%%%%%%%
\subsection{Multicontact and premulticontact structures}

Let $P$ be a manifold with $\dim{P}=m+N$ and $N\geq m\geq 1$, 
and two forms $\Theta,\omega\in\df^m(P)$ with constant rank.
These forms play different roles:
one of them, $\omega$, is a ``reference form'', 
while the other, $\Theta$, is the one that gives the structure that we want to introduce, properly said.

First, given a regular distribution ${\cal D}\subset\Tan P$, consider $\Gamma({\cal D})$, the set of sections of ${\cal D}$. For every $k\in\Nat$, let
$$
{\cal A}^k({\cal D}):=\{ \alpha\in\df^k(P)\mid
\inn(Z)\alpha=0\,;\, \forall Z\in\Gamma({\cal D})\} \,;
$$
that is, the set of differential $k$-forms on $P$
vanishing by the vector fields of $\Gamma({\cal D})$. At a point ${\rm p}\in P$, the point-wise version is
$\displaystyle
{\cal A}_{\rm p}^k({\cal D}):=\big\{ \alpha\in\Lambda^k\Tan^*_{\rm p}P\mid
\inn(v)\alpha=0\,;\,  \forall v\in{\cal D}_{\rm p}\big\}$.

\begin{lemma} 
\label{involutive}
If $\mathcal{D}$ is an involutive distribution and $\alpha\in \mathcal{A}^k(\mathcal{D})$, we have
$$
i(X)i(Y)\d\alpha=0\,,
$$
for every $X,Y\in\Gamma(\mathcal{D})$.
\end{lemma}
\begin{proof}
Since $\mathcal{D}$ is involutive, we have that $[X,Y]\in\Gamma(\mathcal{D})$ for every $X,Y\in\Gamma(\mathcal{D})$. Then
$$
0=\Lie_X(i(Y)\alpha)=i([X,Y])\alpha+i(Y)\d (i(X)\alpha)+i(Y)i(X)\d\alpha=
i(Y)i(X)\d\alpha\,.
$$
\end{proof}

For a form $\alpha\in\df^k(P)$, with $k>1$, the `$1$-$\ker$ of $\alpha$' will be simply denoted as 
$\ker\alpha$; that is,
$\ker\alpha =\{ Z\in\vf(P)\mid \inn(Z)\alpha=0\}$. With this in mind, the above definition of ${\cal A}^k({\cal D})$ can be written as
$$
{\cal A}^k({\cal D})=\big\{ \alpha\in\df^k(P) \mid 
\Gamma(\cal D)\subset\ker\alpha \big\} \,.
$$
Then, for a pair $(\Theta,\omega)$ we define:

\begin{definition}
\label{def:reeb_dist}
% Let $\Gamma(\ker\omega)$ be the space of sections of the subbundle $\ker\omega$.
%${\cal K}_\omega\subset\Tan P$ the distribution generated by $\ker\omega$.
The \textbf{Reeb distribution} associated to
the pair $(\Theta,\omega)$ is the distribution
${\cal D}^{\mathfrak{R}}\subset\Tan P$ defined, at every point ${\rm p}\in P$, as
\[
{\cal D}_{\rm p}^{\mathfrak{R}}=\big\{ v\in(\ker\omega)|_{\rm p}\ \mid\ \inn(v)\dd\Theta_{\rm p}\in {\cal A}^{m}_{\rm p}(\ker\omega)\big\} \,,
\]
and $\displaystyle{\cal D}^{\mathfrak{R}}=\bigcup_{{\rm p}\in P}{\cal D}_{\rm p}^{\mathfrak{R}}$.
The set of sections of the Reeb distribution is denoted by $\mathfrak{R}:=\Gamma({\cal D}^{\mathfrak{R}})$,
and its elements $R\in\mathfrak{R}$ are called \textbf{Reeb vector fields}.
Then, if $\ker \omega$ has constant rank,
\beq
\label{Reebdef}
\mathfrak{R}=\big\{ R\in\Gamma(\ker\omega)\,|\, \inn(R)\dd\Theta\in {\cal A}^{m}(\ker\omega)\big\} \,.
\eeq
\end{definition}

Note that  $\ker\omega\cap\ker\dd\Theta\subset\mathfrak{R}$.

\begin{lemma}
\label{lem:reebinvolutive}
If $\omega$ is a closed form and has constant rank, 
then $\mathfrak{R}$ is involutive.
\end{lemma}
\begin{proof}
The distribution $\ker \omega$ is involutive because $\omega$ is closed.
For every $R_1,R_2\in \mathfrak{R}$, we have that $i(R_1)\d\Theta, i(R_2)\d\Theta\in \mathcal{A}^m(\ker\omega)$
and, as a consequence, their exterior differentials vanish by the action of two vector fields of $\Gamma(\ker\omega)$. 
Then, for every $Y\in \Gamma(\ker\omega)$,
\begin{align*}
i(Y)\left[i([R_1,R_2])\d \Theta\right]&=i(Y)\left[\Lie_{R_1}(i(R_2)\d\Theta)-i(R_2)\Lie_{R_1}\d\Theta\right]
\\ &=i(Y)\left[i(R_1)\d i(R_2)\d\Theta-i(R_2)\d i(R_1)\d\Theta\right]=0 \,.
\end{align*}
\end{proof}

\begin{definition}
\label{premulticontactdef}
The pair $(\Theta,\omega)$ is a \textbf{premulticontact structure} if $\omega$ is a closed form and, for $0\leq k\leq N-m$,
we have that:
\begin{enumerate}[{\rm (1)}]
\item\label{prekeromega}
%$\rk\omega=m$; that is,
$\rk\ker\omega=N$.
\item\label{prerankReeb}
$\rk{\cal D}^{\mathfrak{R}}=m+k$.
\item\label{prerankcar}
$\rk\left(\ker\omega\cap\ker\Theta\cap\ker\d\Theta\right)=k$.
\item \label{preReebComp}
${\cal A}^{m-1}(\ker\omega)=\{\inn(R)\Theta\mid R\in \mathfrak{R}\}$,
%\item
%$\df^{m-1}({\cal K}^0)=\{ \inn(R)\Theta\, ;\, R\in\mathfrak{R}\}$.
\end{enumerate}
Then, the triple $(P,\Theta,\omega)$ is said to be
a \textbf{premulticontact manifold}
and $\Theta$ is called a \textbf{premulticontact form} on $P$. 
The distribution
$\mathcal{C}\equiv\ker\omega\cap\ker\Theta\cap\ker\d\Theta$
is called the \textbf{characteristic distribution} of $(P,\Theta,\omega)$.

If $k=0$, the pair $(\Theta,\omega)$ is a \textbf{multicontact structure},
$(P,\Theta,\omega)$ is a \textbf{multicontact manifold}
and, in this situation, $\Theta$ is said to be a \textbf{multicontact form} on $P$.
\end{definition}

From now on, we will write `(pre)multicontact' to refer interchangeably to both situations
(multicontact and premulticontact),
and write `multicontact' or `premulticontact' to distinghish each of them in particular.

In \cite{Vi-2015}, L. Vitagliano introduces higher codimensional versions of contact manifolds that he calls `multicontact manifolds'. 
This notion is different from ours since, in \cite{Vi-2015},
a multicontact manifold is a manifold equipped with a maximally non-integrable distribution of higher codimension, which is called a multicontact structure. 
However, here we are considering differential forms, not distributions. 
The distinction is just the same as in contact geometry, 
where one can consider a $1$-codimensional contact distribution on a $(2n+1)$-dimensional manifold or a contact $1$-form. 
Both descriptions are related but do not coincide.
Our interest in using differential forms relies on the fact that we seek
a suitable geometric framework in order to write the Lagrangian and Hamiltonian field equations
for classical field theories.

\begin{lemma}
\label{lem:distPre}
The characteristic distribution of a (pre)multicontact manifold $(P,\Theta,\omega)$ is involutive and
$$
\ker\omega\cap\ker\Theta\cap\ker\d\Theta={\cal D}^{\mathfrak{R}}\cap\ker \Theta\,.
$$
\end{lemma}
\begin{proof} Consider the vector fields $X,Z\in \Gamma (\ker\omega\cap\ker\Theta\cap\ker\d\Theta)$ defined on an open set. Then, $[X,Z]\in \Gamma(\ker \omega\cap\ker\d\Theta)$ since $\omega$ and $\d\Theta$ are closed. Moreover, 
$$
\inn([X,Z])\Theta=\Lie_X\inn(Z)\Theta-\inn(Z)\d\inn(X)\Theta-\inn(Z)\inn(X)\d\Theta=0\,.
$$
For the other claim, consider $Y\in\ker\omega\cap\ker\Theta\cap\ker\d\Theta$. We have that $Y\in\ker\omega\cap\ker\d\Theta$, 
thus $Y\in {\cal D}^{\mathfrak{R}}$ and, therefore, $\ker\omega\cap\ker\Theta\cap\ker\d\Theta\subset{\cal D}^{\mathfrak{R}}\cap\ker \Theta\,$. By property (\ref{prerankcar}) in Definition \ref{premulticontactdef}, the rank of the characteristic distribution $\mathcal{C}$ is $k$, thus it needs to be proved that the rank of ${\cal D}^{\mathfrak{R}}\cap \ker \Theta$ is also $k$. At every point ${\rm p}\in P$, consider the linear map
\begin{align*}
 \widetilde\Theta_{{\rm p},{\cal D}^{\mathfrak{R}}} \colon \ {\cal D}_{\rm p}^{\mathfrak{R}} &\longrightarrow \Lambda^{m-1}\Tan^*_{\rm p}P
 \\ 
 v &\longmapsto i(v)\Theta \, .
\end{align*}
 Then, $\Ima(\widetilde \Theta_{{\rm p},{\cal D}^{\mathfrak{R}}})={\cal A}_{\rm p}^{m-1}(\ker\omega)$ has dimension $m$ (by properties (\ref{prekeromega}) and (\ref{preReebComp}) in Definition \ref{premulticontactdef}). Since ${\cal D}_{\rm p}^{\mathfrak{R}}$ has dimension $m+k$ by property (\ref{prerankReeb}), then ${\cal D}^{\mathfrak{R}}\cap \ker \Theta=\ker(\widetilde\Theta_{{\rm p},{\cal D}^{\mathfrak{R}}})$ has dimension $k$.
\end{proof}

Associated to a (pre)multicontact structure, we have the following elements:

\begin{proposition}
\label{sigma}
Given a (pre)multicontact manifold $(P,\Theta,\omega)$,
there exists a unique $1$-form
$\sigma_{\Theta}\in\df^1(P)$ verifying that
$$
\sigma_{\Theta}\wedge\inn(R)\Theta=\inn(R)\dd\Theta \,,\ \text{\rm for every $R\in\mathfrak{R}$} \,.
$$
\end{proposition}
\begin{proof} 
Since $\inn(R)\Theta\in{\cal A}^{m-1}(\ker \omega)$ and 
$\inn(R)\dd\Theta\in{\cal A}^{m}(\ker \omega)$, 
then, if $\sigma_{\Theta}$ exists, it must be an element of ${\cal A}^{1}(\ker \omega)$ and it has $m$ independent coefficients. 
At every point ${\rm p}\in P$, we have that $\{\sigma_{\Theta}\wedge \inn(R)\Theta|_{\rm p}=\inn(R)\dd \Theta|_{\rm p}\}_{R\in \mathfrak{R}}$ is a system of $m+k$ lineal equations with $m$ unknowns, which may not be compatible. 
For every $R\in \mathfrak{R}\subset\Gamma(\ker\omega)$ such that $\inn(R)\Theta=0$, we also have that $\inn(R)\d\Theta=0$ by Lemma \ref{lem:distPre}. 
From property (\ref{preReebComp}) in Definition \ref{premulticontactdef},
we have that the rank of the matrix representing $\inn(R)\Theta\vert_{\rm p}$, in a local basis of $\mathfrak{R}$,
is $m$ and thus the system has a unique solution.
\end{proof}

\begin{definition}
The $1$-form $\sigma_{\Theta}$ is called the \textbf{dissipation form}.
\end{definition}

Using this dissipation form we can define the following operator, which will be used later to set
the field equations in a (pre)multicontact manifold.

\begin{definition}\label{def:bard}
\label{bard}
Let $\sigma_{\Theta}\in\df^1(P)$ be the dissipation form. 
We define the operator
\begin{align*}
\bd:\df^k(P)&\longrightarrow\df^{k+1}(P)
\\
\beta&\longmapsto\bd\beta=\dd \beta+\sigma_{\Theta}\wedge\beta\,.
\end{align*}
\end{definition}

We have that $\bd^2=0$ if, and only if, $\d \sigma_{\Theta}=0$. In this case, it induces a 
{\sl Lichnerowicz--Jacobi cohomology} \cite{LLMP-1999}. 
One consequence is that, locally, there exists a function such that $\sigma_{\Theta} = \d f$ and $\bd\beta=e^{-f}\d (e^{f}\beta)$. 
In this case, we say that the pair $(\Theta,\omega)$ is a \textbf{closed multicontact structure}. 
This is also the condition required in order to consider variational higher-order contact Lagrangian field theories \cite{GLMR-2022} (see Section \ref{vaf}).

Note that the 1-form $\sigma_{\Theta}$ does not only depend on $\Theta$, but also on $\omega$, since the Reeb distribution depends on the form $\omega$.
However, since the volume form is understood and in order to shorten notation, we omit the dependence on $\omega$. 

A premulticontact manifold $(P,\Theta,\omega)$ has three associated distributions: $\ker \omega$, the Reeb distribution ${\cal D}^\mathfrak{R}$, and the characteristic distribution $\mathcal{C}$. 
They are all involutive and are nested: $\mathcal{C}\subset {\cal D}^\mathfrak{R}\subset \ker\omega$. We can use these facts to obtain adapted coordinates.

\begin{theorem}
\label{lem:adaptedcoordgeneral}
Let $P$ be a differentiable manifold and let $D_1,D_2,D_3\subset\T P$ be three involutive distributions on $P$. Suppose that
\begin{enumerate}[{\rm (1)}]
\item $D_3\subset D_2 \subset D_1$,
\item $\rk D_3=r_3$, $\rk D_2=r_3+r_2$ and $\rk D_1=r_3+r_2+r_1$\,,
\item $\dim\, M=r=r_3+r_2+r_1+r_0$\,,
\end{enumerate}
with $r_3,r_2,r_1,r_0\in \Nat^* = \Nat\setminus\{0\}$.
Then, for every ${\rm p}\in M$, there exists a local chart of coordinates $(U;x^1,\ldots,x^r)$, ${\rm p}\in U$, such that:
\begin{enumerate}[{\rm (1)}]
\item $\displaystyle D_3\vert_U=\Big\langle \frac{\partial}{\partial x^1},\ldots,\frac{\partial}{\partial x^{r_3}}\Big\rangle$\,,
\item $\displaystyle D_2\vert_U=\Big\langle \frac{\partial}{\partial x^1},\ldots,\frac{\partial}{\partial x^{r_3}},\frac{\partial}{\partial x^{r_3+1}},\ldots,\frac{\partial}{\partial x^{r_3+r_2}}\Big\rangle$\,,
\item $D_1\displaystyle \vert_U=\Big\langle \frac{\partial}{\partial x^1},\ldots,\frac{\partial}{\partial x^{r_3}},\frac{\partial}{\partial x^{r_3+1}},\ldots,\frac{\partial}{\partial x^{r_3+r_2}},\frac{\partial}{\partial x^{r_3+r_2+1}},\ldots,\frac{\partial}{\partial x^{r_3+r_2+r_1}}\Big\rangle$\,.
\end{enumerate}
\end{theorem}
\begin{proof}
We will follow the ideas contained in the proof of Frobenius theorem given in \cite{book:Narasimhan}.

Let ${\rm p}\in P$, by Frobenius theorem there exists a
chart of coordinates $(W;y^1,\ldots,y^r)$, with ${\rm p}\in W$, such that 
$$
D_1\vert_W=\Big\langle\frac{\partial}{\partial y^1},\ldots,\frac{\partial}{\partial y^{r_3+r_2+r_1}}\Big\rangle\ .
$$
Let $X_1,\ldots,X_{r_3+r_2}$ be local generators of $D_2$ such that
$$
D_1\vert_W=\Big\langle X_1,\ldots,X_{r_3+r_2},\frac{\partial}{\partial y^{r_3+r_2+1}},\ldots,\frac{\partial}{\partial y^{r_3+r_2+r_1}}\Big\rangle\,,
$$
and $[X_i,X_j]=0$ for every $i,j$. For the existence of such generators of $D_2$ see \cite{book:Narasimhan}.
Then we can take the parameters of the local flow of $X_1,\ldots,X_{r_3+r_2}$ as new coordinates, maybe reducing the neighborhood $W$, and obtain a new coordinate system $(V;z^1,\ldots, z^r)$,
with ${\rm p}\in V\subseteq W$, such that
\begin{enumerate}[{\rm (1)}]
\item $\displaystyle D_1\vert_V=\Big<\frac{\partial}{\partial z^1},\ldots,\frac{\partial}{\partial z^{r_3+r_2}},\frac{\partial}{\partial z^{r_3+r_2+1}},\ldots\frac{\partial}{\partial z^{r_3+r_2+r_1}}\Big>\,,$
\item $\displaystyle D_2\vert_V=\Big<\frac{\partial}{\partial z^1},\ldots,\frac{\partial}{\partial z^{r_3+r_2}}\Big>\ .$
\end{enumerate}
Repeating the procedure above with $D_3\subset D_2$, we obtain the required coordinate system in an appropriate neighborhood $U$ of ${\rm p}$.
\end{proof}

\begin{corollary} 
\label{cor:coordenadas}
Around every point ${\rm p}\in P$ of a premulticontact manifold $(P,\Theta,\omega)$, there exists a local chart of {\rm adapted coordinates} 
$(U;x^1,\dots,x^m,u^1\dots,u^{N-m-k},s^{1}\dots,s^{m},w^{1},\dots,w^{k})$ such that
\begin{align*}
\ker\omega\vert_U&=\Big<\frac{\partial}{\partial u^1},\dots, \frac{\partial}{\partial u^{N-m-k}},\frac{\partial}{\partial s^1},\dots, \frac{\partial}{\partial s^m},\frac{\partial}{\partial w^1},\dots, \frac{\partial}{\partial w^{k}}\Big> \,, \\   {\cal D}^\mathfrak{R}\vert_U&=\Big<\frac{\partial}{\partial s^1},\dots, \frac{\partial}{\partial s^m},\frac{\partial}{\partial w^1},\dots, \frac{\partial}{\partial w^{k}}\Big> \,, \\    \mathcal{C}\vert_U&=\Big<\frac{\partial}{\partial w^1},\dots, \frac{\partial}{\partial w^{k}}\Big>\,.
\end{align*}
\end{corollary}
\begin{proof}
It is a straightforward consequence of Theorem \ref{lem:adaptedcoordgeneral}, taking $D_1=\ker\omega$, ${\cal D}^\mathfrak{R}=D_2$ and $\mathcal{C}=D_3$, with $r_0=m$, $r_1=N-m-k$, $r_2=m$, and $r_3=k$.

Observe that for a multicontact manifold,
since ${\cal C}=\{0\}$, there are no coordinates $(w^j)$.
\end{proof}

On these charts, the coordinates $(x^\mu)$
can be chosen in such a way that the form $\omega$ reads 
$\omega\vert_U=\dd x^1\wedge\dots\wedge\dd x^m\equiv\dd^m x$,
and so we shall do henceforth.
Then we denote $\displaystyle\d^{m-1}x_\mu=\inn\Big(\frac{\partial}{\partial x^\mu}\Big)\d^mx$.

Taking into account these results, we can give a local characterization of the Reeb vector fields.

\begin{proposition}
\label{reeblemma}
If $(P,\Theta,\omega)$ is a multicontact manifold, in the above chart of coordinates,
there exists a unique local basis
$\{ R_\mu\}$ of $\mathfrak{R}$ such that 
$$
%\label{reeblocal}
\inn(R_\mu)\Theta=\dd^{m-1}x_\mu\,.
$$
In addition, $[R_\mu,R_\nu] = 0$.
\end{proposition}
\begin{proof}
For every ${\rm p}\in P$ we have the linear map
$\widetilde\Theta_{\rm p}:\Tan_{\rm p}P\rightarrow\Lambda^{m-1}\cT_{\rm p}P$ given by $v\mapsto\inn (v)\Theta_{\rm p}$, for every $v\in\Tan_{\rm p}P$.
For the subset $\mathfrak{R}_{\rm p}\subset\Tan_{\rm p}P$, we have that $\widetilde\Theta_{\rm p}(\mathfrak{R}_{\rm p})={\cal A}_{\rm p}^{m-1}(\ker\omega)$, where both spaces
$\mathfrak{R}_{\rm p}$ and ${\cal A}_{\rm p}^{m-1}(\ker\omega)$ have the same dimension. Then $\widetilde\Theta_{\rm p}$ is an isomorphism and using the above system of coordinates and taking $\{\d^{m-1}x_\mu\vert_{\rm p}\}$
as local basis in ${\cal A}_{\rm p}^{m-1}(\ker\omega)$,
there exists a unique solution to $\inn((R_\mu)_{\rm p})\Theta_{\rm p}=\dd^{m-1}x_\mu\vert_{\rm p}$.
This solution can be extended to the open neighbourhood $U$ of ${\rm p}$ where the local chart of coordinates is defined.
Finally, in $U$:
$$
i([R_\mu,R_\nu])\Theta=\Lie_{R_\mu}(i(R_\nu)\Theta)+i(R_\nu)\d i(R_\mu)\Theta+i(R_\nu)i(R_\mu)\d\Theta=0\,.
$$
Then, $[R_\mu,R_\nu]\in\Gamma(\ker\Theta)$ and $[R_\mu,R_\nu]\in\mathfrak{R}$ because the Reeb distribution is involutive (see Lemma \ref{lem:reebinvolutive}). Therefore, by Lemma \ref{lem:distPre}, $[R_\mu,R_\nu]\in\mathcal{C}=\{0\}$.
\end{proof}

\begin{definition}
\label{reebprop}
The above vector fields $R_\mu\in\mathfrak{R}$ are the \textbf{local Reeb vector fields} of the multicontact manifold  $(P,\Theta,\omega)$
in the chart $U\subset P$.
\end{definition}

From Definition \ref{def:reeb_dist} and, in particular, from \eqref{Reebdef}, we have that there exist local functions $\Gamma_\mu\in \Cinfty(U)$ associated with the basis $\{ R_\mu\}$, which are given by
$$
\inn(R_\mu)\d\Theta = \Gamma_\mu\,\omega\,, \quad \forall \mu\,,
$$
because $\inn(R_\mu)\d\Theta\in{\cal A}_{\rm p}^m(\ker\omega)=\langle\omega\rangle$.
As a consequence, the dissipation form introduced in Proposition \ref{sigma} can be locally expressed on these charts as 
$$
\sigma_{\Theta}=\Gamma_\mu\dd x^\mu \,, 
$$
because $\sigma_{\Theta}\wedge\dd^{m-1}x_\mu=\Gamma_\mu\,\omega=\Gamma_\mu\, \dd^{m}x$, for every $\mu$.

In the premulticontact case, we have the following result.

\begin{proposition}
\label{reeblemma2}
If $(P,\Theta,\omega)$ is a premulticontact manifold, there exist local vector fields
$\{ R_\mu\}$ of $\mathfrak{R}$ such that $\mathfrak{R}=\left<R_\mu\right>+ \mathcal {C}$ and $\inn(R_\mu)\Theta=\dd^{m-1}x_\mu$. They are unique up to a term in the characteristic distribution. Moreover $[R_\mu,R_\nu]\in\Gamma(\mathcal{C})$.
\end{proposition}
\begin{proof}
For every ${\rm p}\in P$ we have the linear map
$\widetilde\Theta_{\rm p}:\Tan_{\rm p}P\rightarrow\Lambda^{m-1}\cT_{\rm p}P$, with $\widetilde\Theta_{\rm p}(\mathcal{C}_{\rm p})=0$. Then, $\widetilde\Theta_{\rm p}(\mathfrak{R}_{\rm p}/\mathcal{C}_{\rm p})={\cal A}_{\rm p}^{m-1}(\ker\omega)$, where both spaces
$\mathfrak{R}_{\rm p}/\mathcal{C}_{\rm p}$ and ${\cal A}_{\rm p}^{m-1}(\ker\omega)$ have the same dimension. Then $\widetilde\Theta_{\rm p}$ is an isomorphism between these spaces and, taking $\{\d^{m-1}x_\mu\vert_{\rm p}\}$
as a local basis in ${\cal A}_{\rm p}^{m-1}(\ker\omega)$,
there exists a unique solution to $\inn((R_\mu)_{\rm p})\Theta_{\rm p}=\dd^{m-1}x_\mu\vert_{\rm p}$, up to a term in $\mathcal{C}_{\rm p}$. This solution can be extended to the open neighbourhood $U$ of ${\rm p}$ where the local chart of coordinates is defined.
The proof of the last assertion is the same as in the above case.
\end{proof}

Note that, using adapted coordinates,
the local Reeb vector fields introduced in Propositions 
\ref{reeblemma} and \ref{reeblemma2}
read $\displaystyle R_\mu=\derpar{}{s^\mu}$.

%%%%%%%%%%%%%%%%%%%%%%%%%%%%%%%%%%%
\subsection{Bundle structures in multicontact and premulticontact manifolds}

Associated to a (pre)multicontact  structure $(\Theta,\omega)$ in a manifold $P$ 
there are the two involutive distributions: $\ker\omega$
and the Reeb distribution ${\cal D}^{\cal R}$,
with ${\cal D}^{\cal R}\subset\ker\omega$.
We can consider the corresponding quotient sets,
$M\equiv P/\ker\omega$ and $\mathcal{E}\equiv P/{\cal D}^{\cal R}$.
From now on we assume:

\noindent\textbf{Assumption:} The quotients $M$ and $\mathcal{E}$ are differentiable manifolds.

\noindent Hence we have the natural projections
$$
\begin{array}{ccccc}
\tau\ \colon & P &\to & M & \\
&(x^\mu,u^I,s^\mu,w^r) &\longmapsto & (x^\mu)\,, & \\ 
\varsigma \ \colon& P &\to & \mathcal{E} &\\
& (x^\mu,u^I,s^\mu,w^r) &\longmapsto & (x^\mu,u^I)\,, & \\
\varepsilon\ \colon& \mathcal{E} &\to & M & \\
& (x^\mu,u^I) &\longmapsto & (x^\mu)\,. &
\end{array}
$$
Furthermore, the form $\omega$ is obviously $\tau$-projectable to a form 
$\omega_{_M}\in\df^m(M)$,
which is a volume form in $M$.
Therefore we can state:

\begin{proposition}
Every (pre)multicontact manifold $(P,\Theta,\omega)$ is locally diffeomorphic 
to a fiber bundle $\tau\colon P\to M$,
where $M$ is an orientable manifold with volume form $\omega_{_M}$, and $\omega=\tau^*\omega_{_M}$.
\end{proposition}

From now on, we assume this  as the canonical model for (pre)multicontact manifolds 
since, in addition, this is the situation which
is interesting in field theories.
Thus, we consider a fiber bundle $\tau\colon P\rightarrow M$, with $\dim M=m$, $\dim P=m+N$,
and such that $M$ is an orientable manifold with volume form $\omega_{_M}\in\df^m(M)$.
Let $\omega=\tau^*\omega_{_M}\in\df^m(P)$.
We always take local coordinates $(x^\mu,z^A)$ in $P$
($1\leq\mu\leq m$, $1\leq A\leq N$),
adapted to the bundle structure, and
such that $\omega=\dd x^1\wedge\dots\wedge\dd x^m\equiv\dd^m x$.
The $\tau$-vertical bundle is defined as
$$
V(\tau)=\bigcup_{{\rm p}\in P}V(\tau_{\rm p})=
\bigcup_{{\rm p}\in P}\{ {\rm v}\in\Tan_{\rm p}P\,\mid\,\Tan_{\rm p}\tau({\rm v})=0\} \,.
$$

Let $\vf^{V(\tau)}(P)$ denote the $\Cinfty(P)$-module of $\tau$-vertical vector fields
and $V(\tau)$ the corresponding {\sl $\tau$-vertical distribution}. 
A form $\alpha\in\df^k(P)$ is $\tau$-semibasic if $\inn(Y)\alpha=0$,
for every $Y\in\vf^{V(\tau)}(P)$. 
Let ${\cal A}^k(V(\tau))$ denote the $\Cinfty(P)$-module of $\tau$-semibasic $k$-forms 
and ${\cal A}_{\rm p}^k(V(\tau))$ the corresponding fiber at ${\rm p}\in P$.
We have that $\Gamma(\ker\omega)=\vf^{V(\tau)}(P)$.

Now, taking these forms $\omega\in\df^m(P)$ and $\Theta\in\df^m(P)$,
Definition \ref{premulticontactdef} adapted
to this context (where condition (1) holds automatically) is: 

\begin{definition}
\label{multicontactbundle}
The pair $(\Theta,\omega)$ is a \textbf{multicontact bundle structure} 
and $(P,\Theta,\omega)$ is said to be
a \textbf{multicontact bundle} if: 
\begin{enumerate}[{\rm(1)}]
\item
$\rk{\cal D}^{\mathfrak{R}}=m$.
\item
$\ker\omega\cap\ker\Theta\cap\ker\d\Theta=\{0\}$.
\item 
${\cal A}^{m-1}(\ker\omega)=\{\inn(R)\Theta\mid R\in \mathfrak{R}\}$.
%\item
%$\df^{m-1}({\cal K}^0)=\{ \inn(R)\Theta\, ;\, R\in\mathfrak{R}\}$.
\end{enumerate}
The pair $(\Theta,\omega)$ is a \textbf{premulticontact bundle structure}
and $(P,\Theta,\omega)$ is said to be
a \textbf{premulticontact bundle} if, for $0< k\leq N-m$,
we have that:
\begin{enumerate}[{\rm(1)}]
\item\label{prerankReebbundle}
$\rk{\cal D}^{\mathfrak{R}}=m+k$.
\item\label{prerankcarbundle}
$\rk\left(\ker\omega\cap\ker\Theta\cap\ker\d\Theta\right)=k$.
\item \label{preReebCompbundle}
${\cal A}^{m-1}(\ker\omega)=\{\inn(R)\Theta\mid R\in \mathfrak{R}\}$,
%\item
%$\df^{m-1}({\cal K}^0)=\{ \inn(R)\Theta\, ;\, R\in\mathfrak{R}\}$.
\end{enumerate}
\end{definition}

In classical field theories we will be specially interested in the situation in which
$P=\mathcal{E}\times\Lambda^{m-1}(\Tan^*M)$,
where $\mathcal{E}\to M$ is a (pre)multisymplectic bundle and, in particular, a jet bundle or a bundle of forms.

\subsection{Multicontact and premulticontact structures of variational type}

Now, we are going to restrict the kind of (pre)multicontact structures we are interested in.
This is motivated by the following fact:
If $(P,\Theta,\omega)$ is a (pre)multicontact manifold,
we can introduce a system of {\sc pde}s associated
with the (pre)multicontact structure.
We want these equations, when expressed in coordinates,
to coincide with those derived from the variational formulation \cite{GLMR-2022}, 
which are also those obtained in the $k$-(co)contact formulation of non-conservative field theories
\cite{GGMRR-2019,GGMRR-2020,GRR-2022,Ri-2022}.
This variational principle is an extension of the {\sl Herglotz variational principle}
for contact mechanical systems.
Nevertheless, in general, these variational equations cannot be derived for any
multicontact structure and then, in order to achieve this,
some additional conditions on the (pre)multicontact structure must be imposed.
So, a more restricted type of multicontact structures called ``variational'' is introduced.
For them, those {\sc pde}s are derived from the variational principle 
using the differential operator $\overline\d$ introduced in Definition \ref{bard}.

Therefore, as we will see in Sections 
\ref{mlf} and \ref{mhf}, in the particular case when $P\to M$ are certain kinds of fiber bundles,
we can formulate Lagrangian and Hamiltonian descriptions for these systems 
and the {\sc pde}s associated
with the ``variational'' multicontact structure are the corresponding
Euler--Lagrange (Herglotz) equations and the
Hamilton--de Donder--Weyl (Herglotz) equations
(we will maintain the usual terminology of the
Lagrangian and the Hamiltonian formalisms of multisymplectic field theories).

In this way we state:

\begin{definition}
If $(P,\Theta,\omega)$ is a (pre)multicontact manifold such that
\beq
\label{varcond}
i(X)i(Y)\Theta = 0 \,, \qquad
\text{\rm for every $X,Y\in\Gamma(\ker\omega)$} \,,
%$X,Y\in\mathfrak{X}^{V(\tau)}(P)$
\eeq
then $(\Theta,\omega)$ is said to be a 
\textbf{variational (pre)multicontact structure} and
$(P,\Theta,\omega)$ is a \textbf{variational (pre)multicontact manifold}.
\end{definition}

The terminology comes from the above comment and from the fact that this condition  \eqref{varcond} is precisely what is imposed to the multisymplectic potential forms in the multisymplectic formulation of field theories in order to ensure that the theory is variational and, hence, it comes from a Lagrangian
(in these cases, $\ker\omega$ is just the vertical distribution on the corresponding bundles).

Now, from the results stated in Theorem \ref{lem:adaptedcoordgeneral}
and Corollary \ref{cor:coordenadas},
we can prove a Darboux-like theorem for this class of (pre)multicontact manifolds:

\begin{theorem}
\label{prop-adapted-coord}
If $(P,\Theta,\omega)$is a variational (pre)multicontact manifold, then there exist local charts of adapted coordinates $(U;x^\mu,u^I,s^\mu,w^r)$ ($1\leq\mu\leq m$, $1\leq I\leq N-m-k$, $1\leq r\leq k$) in $P$ such that the local expression of the (pre)multicontact form $\Theta$ is
\beq
\label{thetacoor}
\Theta\vert_U=H(x^\nu,u^I,s^\nu)\,\d^mx+f_I^\mu(x^\nu,u^J)\,\d u^I\wedge\d^{m-1}x_\mu+
\d s^\mu\wedge\d^{m-1}x_\mu\,.
\eeq
Furthermore, in these coordinates,
\beq
\label{sigmacoord}
\sigma_{\Theta}\vert_U=
\Gamma_\mu\,\d x^\mu=\derpar{H}{s^\mu}\,\d x^\mu\,.
\eeq
\end{theorem}
\begin{proof}
Using the adapted coordinates introduced in Theorem \ref{lem:adaptedcoordgeneral} 
and Corollary \ref{cor:coordenadas},
$$
\ker\omega=\Big<\frac{\partial}{\partial u^I},\frac{\partial}{\partial s^\mu},\frac{\partial}{\partial w^r}\Big> \, , \qquad
{\cal D}^\mathfrak{R}=\Big<\frac{\partial}{\partial s^\mu},\frac{\partial}{\partial w^r}\Big>\, , \qquad 
{\cal C}=\Big\{\derpar{}{w^r}\Big\} \,.
$$
Taking into account condition
\eqref{varcond} and Propositions \ref{reeblemma} and \ref{reeblemma2},
we have that the local expression of $\Theta$ must be, in general,
$$
\Theta=H\,\d^mx+f_I^\mu\,\d u^I\wedge\d^{m-1}x_\mu+g_r^\mu\,\d w^r\wedge\d^{m-1}x_\mu+\d s^\mu\wedge\d^{m-1}x_\mu \, ,
$$
where $H,f_I^\mu,g_r^\mu\in\Cinfty(U)$. Now, bearing in mind \eqref{Reebdef},
we conclude that $H\equiv H(x^\mu,u^I,s^\mu,w^r)$,
$f_I^\mu\equiv f_I^\mu(x^\nu,u^J)$, and
$g_r^\mu\equiv g_r^\mu(x^\nu)$
(in the multicontact case, $k=0$ and hence there are no coordinates $w^r$ and no functions $g_r^\mu$).
Finally, consider Lemma \ref{lem:distPre}; 
in the multicontact case ${\cal C}=\{0\}$ and then ${\cal D}^{\mathfrak{R}}\cap\ker\Theta=\{0\}$;
but, as $\displaystyle {\cal D}^\mathfrak{R}=\Big<\frac{\partial}{\partial s^\mu}\Big>$,
this implies that, in the general case,
$\displaystyle\Big\{\frac{\partial}{\partial w^r}\Big\}\subset\Gamma(\ker\Theta)$
and hence it must be $g_r^\mu(x^\nu)=0$ and \eqref{thetacoor} holds.

Finally, from Proposition \ref{sigma} and \eqref{thetacoor},
we obtain \eqref{sigmacoord}.
\end{proof}

In most physical models of field theory, 
$(x^\mu)$ are spacetime coordinates,
$(u^I)$ are coordinates related to the physical fields,
$(w^r)$ are gauge variables, and
$(s^\mu)$ are the `contact variables' related to `damping' or `dissipative' phenomena
and also to the variational action.

%%%%%%%%%%%%%%%%%%%%%%%%%%%%%%%%%%%
\subsection{(Pre)multicontact variational systems: field equations}

The equations for variational multicontact  and premulticontact bundles 
can be stated using different geometric elements as follows: 

\begin{definition}
\label{mconteq}
Let $(P,\Theta,\omega)$ be a variational (pre)multicontact bundle.
    %The following assertions on a section $\psi\in\Gamma(\tau)$ are equivalent:
\begin{enumerate}[{\rm(1)}]
\item  
The \textbf{(pre)multicontact field equations for sections}
$\psi\colon M\to P$ are
\begin{equation}
\label{sect2}
\inn(\psi^{(m)})(\Theta\circ\psi)=0 \,,\qquad
\inn(\psi^{(m)})(\bd\Theta\circ\psi) = 0\,.
\end{equation}
\item 
The \textbf{(pre)multicontact field equations for $\tau$-transverse, locally decomposable multivector fields} ${\bf X}\in\vf^m(P)$ are
\begin{equation}
\label{vf}
\inn({\mathbf{X}})\Theta=0 \,,\qquad \inn({\bfX})\bd\Theta=0 \,,
\end{equation}
where the condition of $\tau$-transversality is usually imposed 
by requiring that $\inn({\bf X})\omega=1$.
\item
The \textbf{(pre)multicontact field equations for Ehresmann connections}
$\nabla$ on $P\to M$ are
\begin{equation}
\label{Ec}
\inn(\nabla)\Theta=(m-1)\Theta \,,\qquad
\inn(\nabla)\bd\Theta=(m-1)\bd\Theta  \,.
\end{equation}
\end{enumerate}
\end{definition}

The relations among these kinds of field equations 
are given by the following results:

\begin{theorem}
\label{equivteor1}
If ${\bf X}\in\vf^m(P)$ is a representative of
 a class of $\tau$-transverse and integrable $m$-multivector fields $\left\{ {\bf X} \right\} \subset \vf^m(P)$ 
 satisfying the (pre)multicontact field equations for multivector fields \eqref{vf},
 then the integral sections of ${\bf X}$ are solutions to the (pre)multicontact field equations for sections \eqref{sect2}.

Conversely, if $\psi\colon M\to P$ is a solution to the (pre)multicontact field equations for sections \eqref{sect2}, 
then there exist a tubular neighborhood $U\subset P$ of \,$\Ima\psi$
and a $\tau$-transverse and integrable multivector field ${\bf X}\in\vf^m(U)$ 
%which is a representative of a class $\left\{ {\bf X} \right\} \subset \vf^m(U)$,
 such that:
 \begin{enumerate}[{\rm (1)}]
 \item
 $\psi$ is an integral section of ${\bf X}$.
 \item
 ${\bf X}$ is a solution to the (pre)multicontact field equations for multivector fields \eqref{vf}
 on $\Ima\psi$.
 \end{enumerate}
\end{theorem}
\begin{proof}
If $\psi$ is an integral section of a multivector field
${\bf X}\in\vf^m(P)$, then ${\bf X}\circ\psi=\psi^{(m)}$, and hence
if \eqref{vf} holds for ${\bf X}$, then \eqref{sect2} holds for $\psi$, in the corresponding domains.

Conversely, if $\psi\colon  W\subset M \to P$ is a solution to \eqref{sect2} then,
for every $x\in W$ there exists a neighbourhood $W_x\subset W$ of $x$
such that $\psi(W_x)\subset \psi(W)$.
Since $\psi$ is a section and hence $\Ima\psi$ is an embedded submanifold,
the map $\psi\vert_{W_x}$ is an injective immersion
and the map $\Lambda^m(\psi\vert_{W_x})$ defines 
a locally decomposable multivector field 
on $\psi(W_x)\subset P$ which is tangent to $\psi(W_x)$.
Using partitions of unity, this construction can be extended to the whole $W$,
obtaining a multivector field ${\bf X}_{W}$ on $\Ima\psi$
and such that  $\psi$ is an integral section of ${\bf X}_W$, by construction.
Finally, using the local flow of $\tau$-vertical vector fields,
starting from $\Ima\psi$ and ${\bf X}_{W}$
we generate a tubular neighborhood $U$ of $\Ima\psi$
and a multivector field ${\bf X}\in\vf^m(U)$.
Then, since ${\bf X}\circ\psi=\psi^{(m)}$, if equation \eqref{sect2} holds for $\psi$ on $\Ima\psi$, 
then \eqref{vf} holds for ${\bf X}$ on $\Ima\psi$.
\end{proof}

\begin{theorem}
\label{equivteor2}
The (integrable) Ehresmann connections $\nabla$ 
which are the solutions to the (pre)multicontact field equations for Ehresmann connections \eqref{Ec}
are locally associated with classes of (integrable)
$\tau$-transverse, locally decomposable multivector fields $\{{\bf X}\}\subset\vf^m(P)$ which are solutions to the (pre)multicontact field equations for multivector fields \eqref{vf}, and conversely.
\end{theorem}
\begin{proof}
In a chart of adapted coordinates the expressions of 
an Ehresmann connection $\nabla$ on $P$ and of a 
$\tau$-transverse, locally decomposable multivector field ${\bf X}\in\vf^m(P)$, satisfying $\inn({\bf X})\omega=1$, are
\begin{align*}
\nabla&=\d x^\mu\otimes\left(\derpar{}{x^\mu}+F_\mu^I\derpar{}{u^I}+
G_\mu^\nu\derpar{}{s^\nu}+g_\mu^r\derpar{}{w^r}\right) \,,
\\
{\bf X}&=
\bigwedge_{\mu=1}^m \left(\derpar{}{x^\mu}+F_\mu^I\derpar{}{u^I}+
G_\mu^\nu\derpar{}{s^\nu}+g_\mu^r\derpar{}{w^r}\right) \,.
\end{align*}
Now, using Definition \ref{bard},
and the local expressions \eqref{thetacoor} and \eqref{sigmacoord} we have that
\begin{align}
\bd\Theta&=\displaystyle\d(H\,\d^mx+f_I^\mu\,\d u^I\wedge\d^{m-1}x_\mu)-
\Big(\derpar{H}{s^\mu}\,f_I^\mu\,\d u^I+\derpar{H}{s^\mu}\,\d s^\mu\Big)\wedge\d^mx \nonumber
\\ &=\displaystyle
\Big(\derpar{H}{u^I}-\derpar{H}{s^\mu}\,f_I^\mu-\derpar{f_I^\mu}{x^\mu}\Big)\d u^I\wedge\d^mx+\derpar{H}{w^r}\,\d w^r\wedge\d^mx+\derpar{f_I^\mu}{u^J}\,\d u^J\wedge\d u^I\wedge\d^{m-1}x_\mu \,,
\label{bdtheta}
\end{align}
then equations \eqref{vf} give
\begin{align}
0&=\inn({\bf X})\Theta=
G_\mu^\mu+f_I^\mu F_\mu^I+H
\,, \nonumber
\\
0&=\inn({\bf X})\overline\d\Theta=
\derpar{H}{w^r}\,\d w^r+
\Big[\derpar{H}{u^I}-\derpar{H}{s^\mu}\,f_I^\mu-\derpar{f_I^\mu}{x^\mu}+
\Big(\derpar{f_J^\mu}{u^I}-\derpar{f_I^\mu}{u^J}\Big)\,F^J_\mu\Big]\,\d u^I \nonumber
\\ & \quad+
\Big[\Big(\derpar{H}{u^I}-\derpar{H}{s^\mu}\,f_I^\mu-\derpar{f_I^\mu}{x^\mu}\Big)\,F_\nu^I+\derpar{H}{w^r}\,g_\nu^r+ 
\Big(\derpar{f_I^\mu}{u^J}-\derpar{f_J^\mu}{u^I}\Big)\,F_\mu^J F_\nu^I\Big]\,\d x^\nu \,.
\label{2nd}
\end{align}
On the other hand, equations \eqref{Ec} read
\begin{align*}
0&=\inn(\nabla)\Theta-(m-1)\Theta=
G_\mu^\mu+f_I^\mu F_\mu^I+H
\,,
\\
0&=\inn(\nabla)\overline\d\Theta-(m-1)\bd\Theta=
\derpar{H}{w^r}\,\d w^r\wedge\d^mx
\\ &\quad+
\Big[\derpar{H}{u^I}-\derpar{H}{s^\mu}\,f_I^\mu-\derpar{f_I^\mu}{x^\mu}+
\Big(\derpar{f_J^\mu}{u^I}-\derpar{f_I^\mu}{u^J}\Big)\,F^J_\mu\Big]\,\d u^I\wedge\d^mx \,,
\end{align*}
and so, they lead to the same coordinate equations:
\begin{align}
0&=G_\mu^\mu+f_I^\mu F_\mu^I+H \,,  \nonumber \\
0&=\derpar{H}{w^r}\,, \label{3rd} \\
0&=\derpar{H}{u^I}-\derpar{H}{s^\mu}\,f_I^\mu-\derpar{f_I^\mu}{x^\mu}+
\Big(\derpar{f_J^\mu}{u^I}-\derpar{f_I^\mu}{u^J}\Big)\,F^J_\mu \,,  \label{4th}
\end{align}
where equations \eqref{3rd} are just compatibility conditions.
Note that the coefficients on $\d x^\nu$
in \eqref{2nd} vanish identically when 
\eqref{3rd} and \eqref{4th} are used.

Finally, if $\nabla$ and the class $\{ {\bf X}\}$ are locally associated,
$\nabla$ is integrable if, and only if, ${\bf X}$ is integrable,
since the corresponding associated distribution must be involutive.
\end{proof}

As a last result, the field equations for sections can be expressed
in an equivalent way which is analogous to what is commonly used 
to write such equations in the multisymplectic formulation of classical field theories (see \cite{book:Saunders89}):

\begin{proposition}
\label{sectequiv}
The (pre)multicontact field equations for sections \eqref{sect2} are equivalent to
\begin{equation}\label{sect1}
 \psi^*\Theta= 0  \,,\qquad
\psi^*\inn(Y)\bd\Theta= 0 \,, \qquad \text{\rm for every }\ Y\in\vf(P) \,.
\end{equation}
\end{proposition}
\begin{proof}
In a chart of adapted coordinates $(U;x^\mu,u^I,s^\mu,w^r)$
(see Theorem \ref{prop-adapted-coord}), 
for every $Y\in \vf(P)$, for every section $\psi\in\Gamma(\tau)$, and $x\in M$,
we have that
\begin{align*}
Y&=a^\mu\derpar{}{x^\mu}+b^I\derpar{}{u^I}+c^\mu\derpar{}{s^\mu}+g^r\derpar{}{w^r} \,, \\
\psi(x)&=(x^\mu,u^I(x),s^\mu(x),w^r(x)) \,, \\
\psi^{(m)}&=\bigwedge_{\mu=1}^m\left(\derpar{}{x^\mu}+\derpar{u^I}{x^\mu}\derpar{}{u^I}+\derpar{s^\nu}{x^\mu}\derpar{}{s^\nu}+\derpar{w^r}{x^\mu}\derpar{}{w^r}\right) \,.
\end{align*}
Therefore, using the local expressions \eqref{thetacoor}
and \eqref{bdtheta},
equations \eqref{sect2} read
\begin{align*}
0&=\inn(\psi^{(m)})(\Theta\circ\psi)=
\derpar{s^\mu}{x^\mu}+f_I^\mu\,\derpar{u^I}{x^\mu}+H\,,
\\
0&=\inn(\psi^{(m)})(\bd\Theta\circ\psi)=
\derpar{H}{w^r}\,\d w^r+
\Big[\derpar{H}{u^I}-\derpar{H}{s^\mu}\,f_I^\mu-\derpar{f_I^\mu}{x^\mu}+
\Big(\derpar{f_J^\mu}{u^I}-\derpar{f_I^\mu}{u^J}\Big)\,\derpar{u^J}{x^\mu}\Big]\,\d u^I
\\ &\quad+
\Big[\Big(\derpar{H}{u^I}-\derpar{H}{s^\mu}\,f_I^\mu-\derpar{f_I^\mu}{x^\mu}\Big)\,\derpar{u^I}{x^\nu}+\derpar{H}{w^r}\,\derpar{g^r}{x^\nu}+ 
\Big(\derpar{f_I^\mu}{u^J}-\derpar{f_J^\mu}{u^I}\Big)\,\derpar{u^J}{x^\mu} \,\derpar{u^I}{x^\nu}\Big]\,\d x^\mu \,,
\end{align*}
and equations \eqref{sect1} read
\begin{align*}
0&=\psi^*\Theta=
\Big(\derpar{s^\mu}{x^\mu}+f_I^\mu\,\derpar{u^I}{x^\mu}+H\Big)\,\d^mx\,,
\\
0&=\psi^*\inn(Y)\bd\Theta=
\left(\derpar{H}{w^r}\,g^r+  \Big[\derpar{H}{u^I}-\derpar{H}{s^\mu}\,f_I^\mu-\derpar{f_I^\mu}{x^\mu}+
\Big(\derpar{f_J^\mu}{u^I}-\derpar{f_I^\mu}{u^J}\Big)\,\derpar{u^J}{x^\mu}\Big]\,b^I\right.
\\ &\quad+
\left.\Big[\Big(\derpar{H}{u^I}-\derpar{H}{s^\mu}\,f_I^\mu-\derpar{f_I^\mu}{x^\mu}\Big)\,\derpar{u^I}{x^\nu}+\derpar{H}{w^r}\,\derpar{g^r}{x^\nu}+ 
\Big(\derpar{f_I^\mu}{u^J}-\derpar{f_J^\mu}{u^I}\Big)\,\derpar{u^J}{x^\mu} \,\derpar{u^I}{x^\nu}\Big]
\,a^\mu\right)\,\d^mx \,.
\end{align*}
Both of them lead to the same equations:
\begin{align}
0&=\derpar{s^\mu}{x^\mu}+f_I^\mu\,\derpar{u^I}{x^\mu}+H \nonumber \,,
\\
0&=\derpar{H}{w^r} \label{2eq}\,,
\\
0&=\derpar{H}{u^I}-\derpar{H}{s^\mu}\,f_I^\mu-\derpar{f_I^\mu}{x^\mu}+
\Big(\derpar{f_J^\mu}{u^I}-\derpar{f_I^\mu}{u^J}\Big)\,\derpar{u^J}{x^\mu} \label{3eq} \,,
\end{align} 
where equations \eqref{2eq} are not {\sc pde}s, 
but {\sl compatibility conditions\/} which relate 
the component functions of the sections solution.
Note that the equations
$$
0=\Big(\derpar{H}{u^I}-\derpar{H}{s^\mu}\,f_I^\mu-\derpar{f_I^\mu}{x^\mu}\Big)\,\derpar{u^I}{x^\nu}+\derpar{H}{w^r}\,\derpar{g^r}{x^\nu}+ 
\Big(\derpar{f_I^\mu}{u^J}-\derpar{f_J^\mu}{u^I}\Big)\,\derpar{u^J}{x^\mu} \,\derpar{u^I}{x^\nu}
$$
hold identically using \eqref{2eq} and \eqref{3eq}.
\end{proof}

\begin{definition}
A variational (pre)multicontact bundle $(P,\Theta,\omega)$
along with some of the field equations \eqref{sect2}, \eqref{vf} or \eqref{Ec}
is said to be a \textbf{(pre)multicontact system}.
\end{definition}

\begin{remark}
{\rm
In the premulticontact case, in general,
for the premulticontact system $(P,\Theta,\omega)$, 
the field equations for sections $\psi\colon M\to P$, 
multivector fields ${\bf X}\in\vf^m(P)$,
and Ehresmann connections $\nabla$ on $P$
are not compatible on $P$ 
and a constraint algorithm must be implemented in order to find 
a submanifold $P_f\hookrightarrow P$
(when it exists) where there are integrable
distributions whose associated 
multivector fields ${\bf X}$
and Ehresmann connections $\nabla$ are solutions to the premulticontact field equations on $P_f$ and are tangent to $P_f$.
In this situation note that the constraint algorithm and the final solutions
are independent of the Reeb vector fields selected for the premulticontact system, 
as a consequence of the construction of $\sigma_{\Theta}$ (see Proposition \ref{sigma}).}
\end{remark}

\begin{remark}
{\rm
Summarizing, we have introduced different ways of  setting the field equations in classical field theories. 
The equations for sections, written in its two equivalent forms \eqref{sect2} and \eqref{sect1},
give straightforwardly the system of {\sc pde}s to be solved
for describing the behaviour of the system.
On the other hand, the equations for multivector fields
\eqref{vf} and connections \eqref{Ec}
give a more geometrical interpretation of the solutions (as distributions)
that often make it easier to study and characterize qualitative properties of such solutions.
In particular, these geometric characterizations 
are the most suitable in order to apply the constraint algorithms in the case of premulticontact theories. Note that one can write these equations for a general (pre)multicontact system although, if the structure is not variational, the resulting equations may not correspond to those of the Herglotz principle for fields.
}\end{remark}

Finally, we generalize the concept of
{\sl dissipated quantity}
(see \cite{GGMRR-2019,GGMRR-2020}) to this 
(pre)multicontact setting.

\begin{definition}
Let $(P,\Theta,\omega)$ be a (pre)multicontact system 
and let $\mathbf{X}\in\vf^m(P)$ be a solution to the field equations \eqref{vf}.
A differential form $\xi\in\df^{m-1}(P)$ is a \textbf{dissipated quantity} for this system if $\inn({\mathbf{X}})\bd\xi=0.$

In terms of a section $\psi$ solution to the equivalent field equations
\eqref{sect2} or \eqref{sect1},
this condition reads $\psi^*\bd\xi=0$.
\end{definition}

%%%%%%%%%%%%%%%%%%%%%%%%%%%%
\section{Multicontact Lagrangian formalism}
\label{mlf}
%%%%%%%%%%%%%%%%%%%%%%%%%%%%

\subsection{Geometric preliminaries}

Let $\pi:E\rightarrow M$ be a fiber bundle over the spacetime $M$,
where $\dim{M}=m$, $\dim{E}=m+n$,
and hence $\dim{J^1\pi}=m+n+mn$. 
In the Lagrangian setting,
consider the bundle 
$${\cal P}=J^1\pi\times_M\Lambda^{m-1}(\Tan^*M)\simeq J^1\pi\times\Real^m\,,$$
whose natural projections are presented in the next diagram:
$$
\xymatrix{
&\ &  \  &{\cal P}=J^1\pi\times_M\Lambda^{m-1}(\Tan^*M)  \ar[rrd]_{\tau_1}\ar[lld]^{\rho}\ar[ddd]_{\tau}\ &  \   &
\\
&J^1\pi \ar[ddrr]^{\bar{\pi}^1}\ar[d]^{\pi^1}\ & \ & \ & \ &\ar[ddll]_{\tau_o}\Lambda^{m-1}(\Tan^*M) 
\\
&E\ar[drr]^{\pi}\ & \ & \ & \ &  
 \\
&\ & \ &M \ & \ & 
}
$$
If $(x^\mu,y^i)$ are natural coordinates in $E$, then the induced natural coordinates in ${\cal P}$ are $(x^\mu, y^i,y^i_\mu,s^\mu)$ where,
taking $\{\d^{m-1} x_\mu\}$ as the local basis of $\Lambda^{m-1}(\Tan^*M)$, we have that $\xi=s^\mu\,\d^{m-1} x_\mu$, for every $\xi\in\Lambda^{m-1}(\Tan^*M)$.

Note that, since $\Lambda^{m-1}(\Tan^*M)$
is a bundle of forms over $M$, it is endowed with a canonical structure
$\theta\in \df^{m-1}(\Lambda^{m-1}(\Tan^*M))$,
the ``tautological form'',
which is defined as follows:
for every $\bm{\xi} \equiv(x,\xi)\in\Lambda^{m-1}(\Tan^*M)$
and $X^{1}_{\bm{\xi}},\dots,X^{m-1}_{\bm{\xi}}\in\Tan_{\bm{\xi}}(\Lambda^{m-1}(\Tan^*M))$,
$$
\theta_{\bm{\xi}}(X^1_{\bm{\xi}},\cdots,X^{m-1}_{\bm{\xi}}):=
\xi\Big(\Tan_{\bm{\xi}}\tau_o(X^1_{\bm{\xi}}),\dots,\Tan_{\bm{\xi}}\tau_o(X^{m-1}_{\bm{\xi}})\Big) \,.
$$
Its local expression in natural coordinates is $\theta=s^\mu\,\d^{m-1}x_\mu$.

\begin{definition}
The \textbf{canonical action form} is the differential form $\overline{S}\in\df^{m-1}({\cal P})$ defined
as 
$$
\overline{S}:=\tau_1^*\theta \,,
$$
or, what is equivalent,
at every point ${\rm p}\in {\cal P}$,
\[
\overline{S}_{\rm p}(X^1_{\rm p},\cdots,X^{m-1}_{\rm p}):=
\tau_1({\rm p})_{\tau({\rm p})}(\Tan_{\rm p}\tau(X^1_{\rm p}),\dots,\Tan_{\rm p}\tau(X^{m-1}_{\rm p}))\,,\
\text{\rm for every $X^{1}_{\rm p},\dots,X^{m-1}_{\rm p}\in\Tan_{\rm p}{\cal P}$} \,.
\]
\end{definition}

Note that every section $\bm{\psi}:M\rightarrow {\cal P}$
of $\tau$ defines the $(m-1)$-form $\tau_1\circ\bm{\psi}\in\Lambda^{m-1}(\Tan^*M)$
and then $\bm{\psi}^*\overline{S}=\tau_1\circ\bm{\psi}$. It is also immediate to check that $\overline S$ is a $\tau$-semibasic form, whose expression in coordinates is 
\[
\overline{S}=s^\mu\,\d^{m-1} x_\mu\,.
\]
The terminology is justified because, as we will see in Section \ref{vaf}, this form $\overline S$ is closely related to the action of the system:
in fact, $\d\overline{S}$ is the {\sl Lagrangian action} that appears in the {\sl action functional} \eqref{Lagact}
(see also \eqref{SLag})
%(in fact, its `projection' by $\tau_1$ is essentially the {\sl Lagrangian action} that appears in the {\sl action functional\/}), 
and it is also related with the extended contact structures presented in \cite{LGL-2021}.

\begin{definition}
\label{de652}
Let $\bm{\psi}\colon M\rightarrow {\cal P}$ be a section of the projection $\tau$.
Then $\bm{\psi}$ is a \textbf{holonomic section} in ${\cal P}$ if
the section $\psi:=\rho\circ\bm{\psi}\colon M\to J^1\pi$
is holonomic in $J^1\pi$.
We also say that $\bm{\psi}$ is the
\textbf{canonical prolongation} 
 of $\psi$ to ${\cal P}$.

Then, we can write $\bm{\psi}=(\psi,s)=(j^1\phi,s)$, where
$s\colon M\to\Lambda^{m-1}(\Tan^*M)$ 
is a section of the projection
$\tau_0:\Lambda^{m-1}(\Tan^*M)\to M$.
\end{definition}

\begin{definition}
An $m$-multivector field
$\bm{\Gamma}\in\vf^m({\cal P})$ is a \textbf{second-order partial differential equation} (or \textsc{sopde}) in ${\cal P}$ if
\begin{enumerate}[{\rm(1)}]
\item
it is $\tau$-transverse,
\item
it is integrable,
\item
the multivector field 
${\bf X}:=\Lambda^m\Tan\rho\circ \bm{\Gamma}$, which is obviously integrable and $\bar\pi^1$-transverse,
is a {\sc sopde} in $J^1\pi$.
\end{enumerate}
An Ehresmann connection
$\nabla$ in ${\cal P}$ is a \textbf{second-order partial differential equation} (or \textsc{sopde}) in ${\cal P}$ if
\begin{enumerate}[{\rm(1)}]
\item
it is integrable,
\item
the natural restriction of $\nabla$ to $J^1\pi$
is a {\sc sopde} in $J^1\pi$.
\end{enumerate}

\end{definition}
            
The local expression of a {\sc sopde}
multivector field in ${\cal P}$
verifying the transversality condition $\inn(\bm{\Gamma})\omega=1$ is 
$$
%\label{localsode2}
\bm{\Gamma}=\bigwedge^m_{\mu=1}
\Big(\derpar{}{x^\mu}+y^i_\mu\frac{\displaystyle\partial} {\displaystyle
\partial y^i}+\Gamma_{\mu\nu}^i\frac{\displaystyle\partial}{\displaystyle \partial y^i_\nu}+g^\nu_\mu\,\frac{\partial}{\partial s^\nu}\Big)\,.
$$
On the other hand, the local expression of a {\sc sopde} connection is
$$
%\label{localsode2}
\nabla=\d x^\mu\otimes
\Big(\derpar{}{x^\mu}+y^i_\mu\frac{\displaystyle\partial} {\displaystyle
\partial y^i}+\Gamma_{\mu\nu}^i\frac{\displaystyle\partial}{\displaystyle \partial y^i_\nu}+g^\nu_\mu\,\frac{\partial}{\partial s^\nu}\Big)\,.
$$
As usual, multivector fields and connections in ${\cal P}$ which have these local expressions but are not integrable are called {\sl semi-holonomic}.

A straightforward consequence of
the above definitions is that
$\bm{\Gamma}\in\vf^m({\cal P})$ and $\nabla$ are {\sc sopde}s in ${\cal P}$
if, and only if, their integral sections are holonomic in ${\cal P}$.

Finally, since ${\cal P}=J^1\pi\times_M\Lambda^{m-1}(\Tan^*M)$,
%given a volume form in $M$, we can construct 
the canonical endomorphism ${\rm J}$ of $J^1\pi$
can be extended to ${\cal P}$ in a natural way and has the same coordinate expression. 
Denoting this extension with the same notation ${\rm J}$,
in natural coordinates
$\displaystyle {\rm J}=\left(\d y^i-y^i_\mu\d x^\mu\right)\otimes
\derpar{}{y^i_\nu}\otimes\derpar{}{x^\nu}$.

%%%%%%%%%%%%%%%%%%%%%%%%%%%%%%%%%%%
\subsection{(Pre)multicontact Lagrangian systems}

Now we can state the Lagrangian formalism of field theories with dissipation in the multicontact setting.

A \textbf{Lagrangian density} is a $\tau$-semibasic form $\mathcal{L}\in\df^m({\cal P})$.
If $\omega_{_M}$ is the volume form in $M$,
we have that $\Lag=L\,\tau^*\omega_{_M}$,
where $L\in\Cinfty({\cal P})$ is the
\textsl{Lagrangian function} associated to $\Lag$.

\begin{definition}
    \label{lagrangean}
The \textbf{Lagrangian form} associated to $\L$ is the form
\[
\Theta_{\Lag}=-\inn({\rm J})\d\mathcal{L}-\mathcal{L}+\d \overline{S}\in\df^m({\cal P}) \,,
\]
and then \
$\bd\Theta_\L=\d\Theta_\L+\sigma_{\Theta_\Lag}\wedge\Theta_\L\,.$
\end{definition}

In natural coordinates, the expression of the form
$\Theta_\L$ is just \eqref{thetacoor1},
and the local function
$\displaystyle E_\Lag:=\frac{\partial L}{\partial y^i_\mu}y^i_\mu-L$
is called the \textbf{Lagrangian energy} associated with $L$.
Therefore, $\displaystyle \sigma_{\Theta_\Lag}=\derpar{E_\Lag}{s^\mu}\,\d x^\mu$.

\begin{remark}{\rm
The (pre)multicontact form $\Theta_\L$ in ${\cal P}$
can also be obtained in an equivalent way
which is based on using the multisymplectic formalism for Lagrangian field theories
(see Section \ref{nmft}):
If we take the restriction of the Lagrangian function 
$L\in\Cinfty({\cal P})$
 to the fibers of the projection $\tau_1$
 (it is obtained considering $L$ with $s^\mu$ `freezed'),
 as ${\cal P}=J^1\pi\times_M\Lambda^{m-1}(\Tan^*M)$, these fibers are identified with
 $J^1\pi$, and hence this restricted function is $L_s\in\Cinfty(J^1\pi)$.
 Therefore we can construct
 the Poincar\'e--Cartan $m$-form
 $\Theta_{\Lag_s}\in\df^m(J^1\pi)$
 associated with the Lagrangian density $\Lag_s=L_s\,\bar\pi^{1*}\omega$, which has local expression
$$
\Theta_{\mathcal{L}_s}=
\frac{\partial L_s}{\partial y^i_\mu}\,\d y^i\wedge\d^{m-1}x_\mu -\left(\frac{\partial L_s}{\partial y^i_\mu}y^i_\mu-L_s\right)\d^m x \,.
$$}
%Then it is immediate to obtain that:}
\end{remark}

\begin{proposition}
\label{eqthetalag}
The Lagrangian form associated with $\L$ is
$\Theta_{\mathcal{L}}=-\rho^*\Theta_{\mathcal{L}_s}+\d\overline S$.
\end{proposition}

Now, consider the extended and the restricted multimomentum bundles 
${\cal M}\pi$ and $J^{1*}\pi$ introduced in Section \ref{nmft}.
We can construct the extended and the restricted Legendre maps
$\widetilde{FL}_s\colon J^1\pi\to {\cal M}\pi$,
and $FL_s:=\kappa\circ\widetilde{FL}_s\colon J^1\pi\to J^{1*}\pi$
associated with this ``restricted'' Lagrangian function
$L_s\in\Cinfty(J^1\pi)$
whose local expressions coincide with \eqref{FL1}.
Remember that $L_s$ is a regular Lagrangian in $J^1\pi$ if
 $FL_s$ is a local diffeomorphism or, what is equivalent,
 if the Hessian matrix
$\displaystyle\Big(\frac{\partial^2L_s}{\partial y^i_\mu\partial y^j_\nu}\Big)$
is everywhere regular (nondegenerate),
and $L_s$ is {\sl hyperregular} when $FL_s$ is a global diffeomorphism.

These considerations lead us to introduce the manifold
 ${\cal P}^*:=J^{1*}\pi\times_M\Lambda^{m-1}(\Tan^*M)$,
which has natural coordinates $(x^\mu,y^i,p_i^\mu,s^\mu)$.
 
\begin{definition}
\label{Legmap}
The \textbf{Legendre map}  associated with the Lagrangian function $L\in\Cinfty({\cal P})$
is the map
${\cal F}L\colon {\cal P}\to {\cal P}^*$
given by ${\cal F}L:=(FL_s,{\rm Id}_{\Lambda^{m-1}(\Tan^*M)})$,
\end{definition}

The Legendre map is locally given by
$\displaystyle{\cal F}L(y^i,y^i_\mu,s^\mu)=\Big(y^i,\frac{\partial L}{\partial y^i_\mu},s^\mu\Big)$.

\begin{proposition}
\label{Prop-regLag}
For a Lagrangian function $L\in\Cinfty({\cal P})$, the following conditions are equivalent:
\begin{enumerate}[{\rm (1)}]
\item
The Legendre map
${\cal F}L$ is a local diffeomorphism.
\item
The Hessian matrix
$\displaystyle (W_{ij}^{\mu\nu})= 
\bigg(\frac{\partial^2L}{\partial y^i_\mu\partial y^j_\nu}\bigg)$
is regular everywhere.
\item
The Lagrangian form $\Theta_\L$ is a multicontact form in ${\cal P}$
and $(\Theta_\L,\omega)$ is a multicontact structure.
\end{enumerate}
\end{proposition}
\begin{proof}
The equivalence between (1) and (2) can be easily proved using natural coordinates in ${\cal P}$ and bearing in mind the
local expression of the Legendre map ${\cal F}L$.

For the equivalence with (3), bearing in mind Definition \ref{multicontactbundle} 
we observe that the multicontact structure is characterized by the fact that 
$\Gamma({\cal C})=\{ 0\}$ or, what is equivalent,
$\rk{\cal D}^{\mathfrak{R}}=m$,
its lowest possible value,
and this means that $\rk\d\Theta_\L$
reaches its maximum value.
But, from Proposition \ref{eqthetalag},
we have that $\d\Theta_\L=\rho^*\d\Theta_{\L_s}$,
and then $\rk\d\Theta_\L=\rk\d\Theta_{\L_s}=\dim\,J^1\pi$,
namely $\d\Theta_{\L_s}$ is a multisymplectic form in $J^1\pi$.
As it is well-known from the multisymplectic formulation of classical field theories,
this happens if, and only if, $L_s$ is a regular Lagrangian in $J^1\pi$, for every $s=(s^\mu)\in\Lambda^{m-1}(\Tan^*M)$,
that is $\displaystyle 
\bigg(\frac{\partial^2L_s}{\partial y^i_\mu\partial y^j_\nu}\bigg)$ 
is regular everywhere in $J^1\pi$, for every $s$,
and since $\displaystyle 
\frac{\partial^2L}{\partial y^i_\mu\partial y^j_\nu}({\rm p})=
\frac{\partial^2L_s}{\partial y^i_\mu\partial y^j_\nu}(\rho({\rm p}))$, with $s=\tau_1({\rm p})$,
condition (2) holds.
 \end{proof}

\begin{definition}
A Lagrangian function $L\in\Cinfty({\cal P})$ is said to be \textbf{regular} if the equivalent
conditions in Proposition \ref{Prop-regLag} hold.
Otherwise $L$ is a \textbf{singular} Lagrangian.
In particular, 
$L$ is said to be \textbf{hyperregular} 
if ${\cal F}L$ is a global diffeomorphism.
\end{definition}

As we have seen, $L$ is regular in ${\cal P}$ if, and only if, $L_s$ is regular in $J^1\pi$, for every $s\in\Lambda^{m-1}(\Tan^*M)$.

\begin{remark}
{\rm
It is important to point out that
non-regular Lagrangians can induce premulticontact structures but also structures which are neither multicontact nor premulticontact.
For example, 
the Lagrangian $\displaystyle L=\sum_{i=1}^n y^i_\mu s^\mu$
yields a structure $(\Theta_\L,\omega)$
which has no Reeb distribution associated to it.
}\end{remark}

\begin{definition}
The premulticontact bundle $({\cal P},\Theta_\L,\omega)$ is called a \textbf{(pre)multicontact Lagrangian system}.
\end{definition}

Given a multicontact Lagrangian system $({\cal P},\Theta_\L,\omega)$,
from Lemma \ref{reeblemma} we have that
the Reeb vector fields $(R_\L)_\mu\in\mathfrak{R}_\L\subset\vf({\cal P})$ 
for this system are the unique solutions to
$\inn((R_\L)_\mu)\Theta=\d^{m-1}x_\mu$.
Then, since $L$ is regular, there exists the inverse 
$(W^{ij}_{\mu\nu})$ of the Hessian matrix,
namely $\displaystyle W^{ij}_{\mu\nu}\frac{\partial^2L}{\partial y^j_\nu \partial y^k_\gamma}=\delta^i_k\delta^\gamma_\mu$,
and a simple calculation in coordinates leads to
$$
(R_\L)_\mu=\frac{\partial}{\partial s^\mu}-W^{ji}_{\gamma\nu}\frac{\partial^2\L}{\partial s^\mu\partial y^j_\gamma}\,\frac{\partial}{\partial y^i_\nu} \,.
$$
Therefore, bearing in mind Proposition \ref{sigma} and \eqref{thetacoor1},
we see that
\beq\displaystyle \sigma_{\Theta_\L}=-\derpar{L}{s^\mu}\,\d x^\mu\,.
\label{sigmaL}
\eeq
If $({\cal P},\Theta_\L)$ is a premulticontact Lagrangian system, the Reeb vector fields are not uniquely determined from the equation $\inn((R_\L)_\mu)\Theta=\d^{m-1}x_\mu$.

Note that, in general, the natural coordinates in ${\cal P}$ are not adapted coordinates for the
(pre)multicontact structure $(\Theta_\L,\omega)$.

%%%%%%%%%%%%%%%%%%%%%%%%%%%%%%%%%%%%%%%%%%%
\subsection{The (pre)multicontact Lagrangian field equations}

Bearing in mind Definition \ref{mconteq}, Proposition \ref{sectequiv},
and Theorems \ref{equivteor1} and \ref{equivteor2}, we can define:

\begin{definition}
Let $({\cal P},\Theta_\Lag,\omega)$ be a (pre)multicontact Lagrangian system.
\begin{enumerate}[{\rm(1)}]
\item
The \textbf{(pre)multicontact Lagrangian equations} for holonomic sections
$\bm{\psi}\colon M\to {\cal P}$ are
\begin{equation}
\label{ELmcontact}
\inn(\bm{\psi}^{(m)})(\Theta_\Lag\circ\bm{\psi})=0 \,,\qquad
\inn(\bm{\psi}^{(m)})(\overline\d\Theta_\Lag\circ\bm{\psi})=0 \,.
\end{equation}
or equivalently
\begin{equation}
\label{ELmcontact2}
\bm{\psi}^*\Theta_\Lag= 0 \,,\qquad
\bm{\psi}^*\inn(Y)\bd\Theta_\Lag= 0 \,, \qquad \text{\rm for every }\ Y\in\vf({\cal P}) \,.
\end{equation}
\item
The \textbf{(pre)multicontact Lagrangian equations for $\tau$-transverse, locally decomposable multivector fields}
${\bf X}_\L\in\vf^m({\cal P})$ are
\beq 
\inn({\mathbf{X}_\Lag})\Theta_\Lag=0 \,,\qquad \inn({\bfX}_\Lag)\bd\Theta_\Lag=0 \,,
\label{fieldLcontact}
\eeq
where the condition of $\tau$-transversality is imposed 
by requiring that $\inn({\bf X}_\L)\omega=1$.

An $m$-multivector field solution to these equations is called a
\textbf{Lagrangian multivector field}.
\item
The \textbf{(pre)multicontact Lagrangian equations for Ehresmann connections}
$\nabla_\Lag$ on ${\cal P}\to M$ are
\begin{equation}
\label{EcL}
\inn(\nabla_\L)\Theta_\L=(m-1)\Theta_\L \,,\qquad
\inn(\nabla_\L)\bd\Theta_\L=(m-1)\bd\Theta_\L  \,.
\end{equation}
An Ehresmann connection solution to these equations is called a
\textbf{Lagrangian connection}.
\end{enumerate}
\end{definition}

\begin{proposition}
Let $({\cal P},\Theta_\L,\omega)$ be a multicontact (i.e., regular) Lagrangian system.
Then:
\begin{enumerate}[\rm (1)]
\item
The multicontact Lagrangian field equations
for multivector fields \eqref{fieldLcontact} 
and for Ehresmann connections \eqref{EcL}
have solutions on ${\cal P}$.
The solutions are not unique if $m>1$.
\item
The Lagrangian $m$-multivector fields ${\bf X}_\L$
solution to equations \eqref{fieldLcontact} 
and the corresponding Ehresmann connections
$\nabla_\L$ in ${\cal P}$
which are associated with the classes $\{{\bf X}_\L\}$ 
and are solutions to \eqref{EcL},
are semi-holonomic.
\item
In addition, if ${\bf X}_\L$ and $\nabla_\L$ are semi-holonomic and integrable solutions, namely {\sc sopde}s, their integral sections are solutions to the multicontact Euler--Lagrange field equations \eqref{ELmcontact} or \eqref{ELmcontact2}.

\noindent In this case, these {\sc sopde}s ${\bf X}_\L$
and $\nabla_\L$ are called the \textbf{Euler--Lagrange multivector fields}
and \textbf{connections}
associated with the Lagrangian function $L$.
\end{enumerate}
\end{proposition}
\begin{proof}
In a natural chart of coordinates of ${\cal P}$, 
for an $m$-multivector field
$$\displaystyle
{\bf X}_\L= \bigwedge_{\mu=1}^m
\bigg(\derpar{}{x^\mu}+(X_\L)_\mu^i\frac{\displaystyle\partial}{\displaystyle
\partial y^i}+(X_\L)_{\mu\nu}^i\frac{\displaystyle\partial}{\displaystyle\partial y^i_\nu}+(X_\L)_\mu^\nu\,\frac{\partial}{\partial s^\nu}\bigg)\in\vf^m({\cal P}) \,,
$$
or for the Ehresmann connection $\nabla_\L$
in ${\cal P}$ associated with the class $\{{\bf X}_\L\}$,
$$\displaystyle
\nabla_\L= \d x^\mu\otimes
\bigg(\derpar{}{x^\mu}+(X_\L)_\mu^i\frac{\displaystyle\partial}{\displaystyle
\partial y^i}+(X_\L)_{\mu\nu}^i\frac{\displaystyle\partial}{\displaystyle\partial y^i_\nu}+(X_\L)_\mu^\nu\,\frac{\partial}{\partial s^\nu}\bigg) \,,
$$
bearing in mind Definition \ref{bard} and the local expressions \eqref{thetacoor1} and \eqref{sigmaL} we have that
$$
\bd\Theta_\L=
\d\left(-\frac{\partial L}{\partial y^i_\mu}\d y^i\wedge\d^{m-1}x_\mu +\Big(\frac{\partial L}{\partial y^i_\mu}y^i_\mu-L\Big)\d^m x\right)
+\left(\derpar{L}{s^\mu}\frac{\partial L}{\partial y^i_\mu}\d y^i
-\derpar{L}{s^\mu}\d s^\mu\right)\wedge\d^mx
\,,
$$
and equations \eqref{fieldLcontact} and \eqref{EcL} lead to
\begin{align}
0 &=
\displaystyle L + 
\frac{\partial L}{\partial y^i_\mu}\Big((X_\L)_\mu^i-y^i_\mu\Big)-(X_\L)_\mu^\mu\,,
\label{A-E-L-eqs4}
\\
0 &=
\displaystyle \Big((X_\L)_\mu^j-y^j_\mu\Big)
\frac{\partial^2L}{\partial v^j_\mu\partial s^\nu} \,,
\label{A-E-L-eqs2}
\\
\displaystyle 0&=\Big((X_\L)_\mu^j-y^j_\mu\Big)
\frac{\partial^2L}{\partial v^j_\mu\partial x^\nu} \,,
\label{A-E-L-eqs0}
\\
0 &=
\displaystyle \Big((X_\L)_\mu^j-y^j_\mu\Big)
\frac{\partial L}{\partial y^i_\nu\partial y^j_\mu}
\label{A-E-L-eqs1} \,,
\\
0 &=
\displaystyle
\Big((X_\L)_\mu^j-y^j_\mu\Big)
\frac{\partial^2 L}{\partial y^i\partial y^j_\mu}
+\frac{\partial L}{\partial y^i}- \parderr{L}{x^\mu}{y_\mu^i}
-\frac{\partial^2L}{\partial s^\nu\partial y^i_\mu}(X_\L)_\mu^\nu
\nonumber
\\ &\quad
-\frac{\partial^2L}{\partial y^j \partial y^i_\mu}(X_\L)_\mu^j
-\frac{\partial^2L}{\partial y^j_\nu\partial y^i_\mu}(X_\L)_{\mu\nu}^j
+\frac{\partial L}{\partial s^\mu}
\frac{\partial L}{\partial y^i_\mu}\,,
\label{A-E-L-eqs3}
\end{align}
and a last group of equations which are identities when they are combined with the above ones.
If $L$ is a regular Lagrangian, equations \eqref{A-E-L-eqs1} give 
\beq
\label{semihol}
y^i_\mu=(X_\L)_\mu^i \,,
\eeq
which are the conditions for the multivector field ${\bf X}_\L$ 
and the connection $\nabla_\L$ to be semi-holonomic.
Then, \eqref{A-E-L-eqs2} and \eqref{A-E-L-eqs0} hold identically, 
and \eqref{A-E-L-eqs4} and \eqref{A-E-L-eqs3} give 
\begin{align*}
(X_\L)_\mu^\mu&= L \,,
\\
\frac{\partial L}{\partial y^i}- \parderr{L}{x^\mu}{y_\mu^i}
-\frac{\partial^2L}{\partial y^j \partial y^i_\mu}y_\mu^j
-\frac{\partial^2L}{\partial s^\nu\partial y^i_\mu}(X_\L)_\mu^\nu
-\frac{\partial^2L}{\partial y^j_\nu\partial y^i_\mu}(X_\L)_{\mu\nu}^j
&=-\frac{\partial L}{\partial s^\mu}
\frac{\partial L}{\partial y^i_\mu} \,.
\end{align*}
These equations have always solution since the Hessian matrix 
$\displaystyle\bigg(\frac{\partial^2L}{\partial y^j_\nu\partial y^i_\mu}\bigg)$ is regular everywhere.

Finally, if these semi-holonomic multivector fields ${\bf X}_\L$ and connections $\nabla_\L$ are integrable, by \eqref{semihol}, they are {\sc sopde}s 
and these last equations transform into
\begin{align}
 \derpar{s^\mu}{x^\mu}&=L\circ{\bm{\psi}} \,,
 \label{ELeqs2}
 \\
\label{ELeqs1}
\frac{\partial}{\partial x^\mu}
\left(\frac{\displaystyle\partial L}{\partial
y^i_\mu}\circ{\bm{\psi}}\right)&=
\left(\frac{\partial L}{\partial y^i}+
\displaystyle\frac{\partial L}{\partial s^\mu}\displaystyle\frac{\partial L}{\partial y^i_\mu}\right)\circ{\bm{\psi}} \,,
\end{align}
which are the coordinate expression of the Lagrangian equations
\eqref{ELmcontact} or \eqref{ELmcontact2}
for the integral sections of ${\bf X}_\L$ and $\nabla_\L$. 
\end{proof}

Of course all these equations are the same as those obtained in the $k$-cocontact formulation of non-conservative field theories \cite{Ri-2022}
and also match those of $k$-contact formalism when the Lagrangian function does not depend on the spacetime variables $x^\mu$ \cite{GGMRR-2020,GRR-2022}.
Furthermore, equation \eqref{ELeqs2} relates the canonical action form with the variational formulation through the Lagrangian density
(see \eqref{Lagact}).
In fact, we have the following.

\begin{corollary}
If $\bm{\psi}$ is a holonomic section such that $\bm{\psi}^\ast\Theta_\L = 0$, we have that
\beq
\label{SLag}
\d(\overline S\circ{\bm{\psi}})=\L\circ{\bm{\psi}}\,.
\eeq
\end{corollary}
\begin{proof}
It is immediate since \eqref{ELeqs2} is the coordinate expression of \eqref{SLag}.
\end{proof}

\begin{remark}\rm
As in the case of premultisymplectic field theories,
when $L$ is not regular and $({\cal P},\Theta_\L,\omega)$
is a premulticontact system, 
the field equations \eqref{ELmcontact},
\eqref{ELmcontact2}, \eqref{fieldLcontact}, or \eqref{EcL} have no
solutions everywhere on ${\cal P}$, in general. 
In the most favourable situations,
these equations have solutions on a submanifold of ${\cal P}$ which is obtained by applying a suitable constraint algorithm.
Nevertheless, solutions to equations \eqref{fieldLcontact} or \eqref{EcL}
are not necessarily {\sc sopde}s and,
as a consequence, if they are integrable, 
their integral sections are not necessarily holonomic;
so this requirement must be imposed as an additional condition.
Hence, the final objective consists in finding the maximal submanifold 
${\cal S}_f$ of ${\cal P}$ where there are 
holonomic distributions whose associated 
Lagrangian multivector fields ${\bf X}_\L$
and connections $\nabla_\L$ are
{\sc sopde} solutions to the premulticontact Lagrangian field equations on ${\cal S}_f$ and are tangent to ${\cal S}_f$.
\end{remark}

%%%%%%%%%%%%%%%%%%%%%%%%%%%%%%%%%%%%%%%%%%%%%%%%%%%%%%%%%%%%%%%%
\section{Multicontact Hamiltonian formalism}
\label{mhf}
%%%%%%%%%%%%%%%%%%%%%%%%%%%%%%%%%%%%%%%%%%%%%%%%%%%%%%%%%%%%%%%%

\subsection{The (hyper)regular case}

Consider a multicontact Lagrangian system $({\cal P},\Theta_\L,\omega)$,
where $L$ is a hyperregular Lagrangian
(the case in which $L$ is regular is the same but changing
${\cal P}$ and ${\cal P}^*$ by the corresponding open sets).
Like in the Lagrangian formalism, we have the diagram
$$
\xymatrix{
&\ &  \  &{\cal P}^*=J^{1*}\pi\times_M\Lambda^{m-1}(\Tan^*M)  \ar[rrd]_{\widetilde\tau_1}\ar[lld]^{\varrho}\ar[ddd]_{\widetilde\tau}\ &  \   &
\\
&J^{1*}\pi \ar[ddrr]^{\bar{\kappa}^1}\ar[d]^{\kappa^1}\ & \ & \ & \ &\ar[ddll]_{\tau_o}\Lambda^{m-1}(\Tan^*M) 
\\
&E\ar[drr]^{\pi}\ & \ & \ & \ &  
 \\
&\ & \ &M \ & \ & 
}
$$

Since ${\cal F}L$ and $FL_s$ are global diffeomorphisms (for every $s$),
we have that $FL_s(J^1\pi)=J^{1*}\pi$
and hence ${\cal F}L({\cal P})={\cal P}^*$.
Furthermore, the canonical action form $\overline S\in\df^m({\cal P})$
is ${\cal F}L$-projectable to an $m$-form in ${\cal P}^*$
which has the same coordinate expression and is denoted
with the same notation $\overline S\in\df^m({\cal P}^*)$.
The same goes for the multicontact form $\Theta_\L\in\df^m({\cal P})$:
there exists an $m$-form $\Theta_{\cal H}\in\df^m({\cal P}^*)$
such that $\Theta_\L={\cal F}L^*\Theta_{\cal H}$,
whose local expression is
\beq
\label{thetaHcoor}
\Theta_{\cal H}=
-p_i^\mu\d y^i\wedge\d^{m-1}x_\mu+H\,\d^m x+\d s^\mu\wedge \d^{m-1}x_\mu \,,
\eeq
where $H\in\Cinfty({\cal P}^*)$ is the \textbf{Hamiltonian function}
defined by $E_\L={\cal F}L^*{\cal H}$ whose local expression is
$H=p^\mu_i({\cal F}L^{-1})^*y_\mu^i-({\cal F}L^{-1})^*L$.

\begin{proposition}
The form $\Theta_{\cal H}$ is a multicontact form
and hence $({\cal P}^*,\Theta_{\cal H},\omega)$
is a multicontact bundle.
\end{proposition}
\begin{proof}
From the coordinate expression \eqref{thetaHcoor}
it is immediate to check that the conditions in Definition \ref{multicontactbundle} hold in this case.
It is also a straightforward consequence of the fact that ${\cal F}L$ is a diffeomorphism,
that $\Theta_\L={\cal F}L^*\Theta_{\cal H}$
and that $\Theta_\L$ is a multicontact form.
\end{proof}

\begin{remark}
\label{rem0}
{\rm
The multicontact form $\Theta_{\cal H}$ in ${\cal P}^*$
can also be obtained in the following alternative way,
which is based on using the multisymplectic formalism for hyperregular Hamiltonian field theories
(see Section \ref{nmft}):
${\cal M}\pi$ is endowed with the canonical  multimomentum Liouville $m$-form 
whose local expression is given in \eqref{Liform}.
We can construct the map  
${\bf h}:=\widetilde{{\cal F}L}_s\circ{\cal F}L_s^{-1}$
(which is a section of the projection $\mathfrak{p}\colon{\cal M}\pi\to J^{1*}\pi$
(see the Diagram \eqref{1stdiag}), and
then we define the Hamilton--Cartan $m$-form
$\Theta_{\bf h}:={\bf h}^*\widetilde\Theta\in\df^m(J^{1*}\pi)$,
whose local expression is given in \eqref{HCform}.
Finally, the multicontact form $\Theta_{\cal H}\in\df^m({\cal P}^*)$
is
$$
\Theta_{\cal H}=-\varrho^*\Theta_{\bf h}+\d\overline S \,,
$$
where we have also denoted $\overline S=\widetilde\tau_1^*\theta$.
This construction can also be done independently of having a starting Lagrangian system, 
simply giving any section of the projection $\mathfrak{p}\colon{\cal M}\pi\to J^{1*}\pi$, called a {\bf Hamiltonian section}. 
}
\end{remark}

\begin{definition}
The triple $({\cal P}^*,\Theta_{\cal H},\omega)$ is called a \textbf{multicontact Hamiltonian system}.
\end{definition}

Observe that the natural coordinates in ${\cal P}^*$ are adapted coordinates for the
multicontact structure $(\Theta_{\cal H},\omega)$.
In particular,
the Reeb vector fields $(R_{\cal H})_\mu\in\mathfrak{R}_{\cal H}$ for this multicontact structure are
$\displaystyle (R_{\cal H})_\mu=\derpar{}{s^\mu}$ and hence, bearing in mind Proposition \ref{sigma} and equation \eqref{thetaHcoor},
we obtain that
\beq
\label{sigmaH}
\sigma_{\cal H}=\derpar{H}{s^\mu}\,\d x^\mu\,.
\eeq

For this multicontact Hamiltonian system, the multicontact field equations
come straightforwardly from Definition \ref{mconteq}, Proposition \ref{sectequiv}, and Theorems \ref{equivteor1} and \ref{equivteor2}.

\begin{definition}
\label{mconteqH}
Let $({\cal P}^*,\Theta_{\cal H},\omega)$ be a multicontact Hamiltonian system.
\begin{enumerate}[{\rm(1)}]
\item  
The \textbf{multicontact Hamilton--de Donder--Weyl equations for sections}
$\bm{\psi}\colon M\to{\cal P}^*$ are
\begin{equation}
\label{sect2H}
\inn(\bm{\psi}^{(m)})(\Theta_{\cal H}\circ\bm{\psi})=0 \,,\qquad
\inn(\bm{\psi}^{(m)})(\bd\Theta_{\cal H}\circ\bm{\psi}) = 0\,.
\end{equation}
or, equivalently,
\begin{equation}
\label{sect1H}
\bm{\psi}^*\Theta_{\cal H}= 0  \,,\qquad
\bm{\psi}^*\inn(Y)\bd\Theta_{\cal H}= 0 \,, \qquad \text{\rm for every }\ Y\in\vf({\cal P}^*) \,.
\end{equation}
\item 
The \textbf{multicontact Hamilton--de Donder--Weyl equations for $\widetilde\tau$-transverse, locally decomposable multivector fields} ${\bf X}_{\cal H}\in\vf^m({\cal P}^*)$ are
\begin{equation}
\label{vfH}
\inn({\mathbf{X}_{\cal H}})\Theta_{\cal H}=0 \,,\qquad \inn({\bfX}_{\cal H})\bd\Theta_{\cal H}=0 \,;
\end{equation}
where the condition of $\widetilde\tau$-transversality is imposed 
by requiring that $\inn({\bf X}_{\cal H})\omega=1$.
\item
The \textbf{multicontact Hamilton--de Donder--Weyl equations for Ehresmann connections}
$\nabla_{\cal H}$ on ${\cal P}^*\to M$ are
\begin{equation}
\label{EcH}
\inn(\nabla_{\cal H})\Theta_{\cal H}=(m-1)\Theta_{\cal H} \,,\qquad
\inn(\nabla_{\cal H})\bd\Theta_{\cal H}=(m-1)\bd\Theta_{\cal H}  \,.
\end{equation}
\end{enumerate}
\end{definition}

All these equations are compatible in ${\cal P}^*$.

In natural coordinates, bearing in mind Definition \ref{bard} and the local expressions \eqref{thetaHcoor} and \eqref{sigmaH},
we have
$$
\overline \d\Theta_{\cal H}=
\d(-p_i^\mu\d y^i\wedge\d^{m-1}x_\mu+H\,\d^m x)+
\Big(\derpar{H}{s^\mu}\,p_i^\mu\,\d y^i-\derpar{H}{s^\mu}\,\d s^\mu\Big)\wedge\d^mx \,.
$$
Then, if
$$
{\bf X}_{\cal H}= 
\bigwedge^\mu\Big(\derpar{}{x^\mu}+ (X_\mu)^i\frac{\partial}{\partial y^i}+
(X_\mu)_i^\nu\frac{\partial}{\partial p_i^\nu}+(X_\mu)^\nu\derpar{}{s^\nu}\Big)
$$
is a multivector field solution to \eqref{vfH}, and 
$$
\nabla_{\cal H}=\d x^\mu\otimes\Big(\derpar{}{x^\mu}+ (X_\mu)^i\frac{\partial}{\partial y^i}+
(X_\mu)_i^\nu\frac{\partial}{\partial p_i^\nu}+(X_\mu)^\nu\derpar{}{s^\nu}\Big)
$$ 
is the Ehresmann connection in ${\cal P}^*$ associated with the class $\{{\bf X}_{\cal H}\}$ 
and it is solution to \eqref{EcH},
these field equations lead to
$$
(X_\mu)^\mu = 
p_i^\mu\,\frac{\partial H}{\partial p^\mu_i}-H \,,\qquad
(X_\mu)^i=\frac{\partial H}{\partial p^\mu_i} \,,\qquad
(X_\mu)^\mu_i= 
-\left(\frac{\partial H}{\partial y^i}+ p_i^\mu\,\frac{\partial H}{\partial s^\mu}\right) \,,
%    \label{coor2}
$$
together with a last group of equations which are identities when the above ones are taken into account.
If $\bm{\psi}(x)=(x^\mu,y^i(x),p^\mu_i(x),s^\mu(x))$
is an integral section of ${\bf X}_{\cal H}$ and $\nabla_{\cal H}$, 
then it is a solution to the equations
\eqref{sect2H} or \eqref{sect1H}, which read
$$
\frac{\partial s^\mu}{\partial x^\mu} = \left(p_i^\mu\,\frac{\partial H}{\partial p^\mu_i}-H\right)\circ\bm{\psi}\,,\qquad 
\frac{\partial y^i}{\partial x^\mu}= \frac{\partial H}{\partial p^\mu_i}\circ\bm{\psi} \,,\qquad
\frac{\partial p^\mu_i}{\partial x^\mu} = 
-\left(\frac{\partial H}{\partial y^i}+ p_i^\mu\,\frac{\partial H}{\partial s^\mu}\right)\circ\bm{\psi} \,.
%\label{coor1}
$$
As in the Lagrangian case,
these equations match those obtained in the 
$k$-contact and $k$-cocontact formulations of non-conservative field theories \cite{GGMRR-2019,GRR-2022,Ri-2022}.

\begin{remark}
{\rm
For multicontact Lagrangian systems $({\cal P},\Theta_\L,\omega)$
and their
associated multicontact Hamiltonian systems
$({\cal P}^*,\Theta_{\cal H},\omega)$,
since ${\cal F}L$ is a diffeomorphism,
the solutions to the Lagrangian field equations are in one-to-one correspondence
to those of the Hamilton--de Donder--Weyl field equations.
} \end{remark}

%%%%%%%%%%%%%%%%%%%%%%%%%%%%%%%%%%%
\subsection{The singular case}

For singular Lagrangians, the existence of
an associated Hamiltonian formalism is not assured, in general,
unless some minimal regularity conditions are assumed.
Then, as it is usual for singular Lagrangian field theories, we introduce the notion of almost-regular Lagrangian.

\begin{definition}
A singular Lagrangian $L\in\Cinfty({\cal P})$ is \textbf{almost-regular} if
\begin{enumerate}[{\rm (1)}]
\item
${\cal P}_0^*:= {\cal F}L({\cal P})$
is a submanifold of ${\cal P}^*$,
\item
${\cal F}L$ is a submersion onto its image,
\item
the fibers ${\cal F}L^{-1}({\rm p})$ are connected submanifolds of ${\cal P}$, for every ${\rm p}\in {\cal P}_0^*$.
\end{enumerate}
\end{definition}

In order to construct a (pre)multicontact structure on ${\cal P}^*_0$
we follow a procedure similar to the one described in Remark \ref{rem0}.
Like in the regular case, we have that
$L\in\Cinfty({\cal P})$ is an almost-regular Lagrangian on ${\cal P}$
if, and only if, 
$L_s\in\Cinfty(J^1\pi)$ is an almost-regular Lagrangian on $J^1 \pi$.
Therefore, following the patterns of the multisymplectic formalism for almost-regular Hamiltonian field theories
(see Section \ref{nmft})
we have the submanifolds $P_0:=FL_s(J^1\pi)$
and $\widetilde P_0:=\widetilde{FL_s}(J^1\pi)$
(see the Diagram \eqref{lastdiag}).
Then, taking $\widetilde{\bf h}:=\widetilde{\mathfrak{p}}^{-1}$,
we define the Hamilton--Cartan $m$-form
$\Theta^0_{\bf h}=(\widetilde\jmath_0\circ\widetilde{\bf h})^*\widetilde\Theta\in\df^m(P_0)$,
and, taking the Hamiltonian
function $H_0\in\Cinfty(P_0)$
such that
$E_\mathcal{L}=FL_{s0}^{\ *}\,H_0$,
the local expression of this form coincides with the one given in \eqref{Thetacero}.
Now, as a consequence of the definition of the
Legendre map ${\cal F}\L$ (Definition \ref{Legmap}),
we have that ${\cal P}_0^*=P_0\times\Lambda^{m-1}(\Tan^*M)$,
and so the diagram depicting the Hamiltonian formalism is
$$
\xymatrix{
&\ &  \  &{\cal P}_0^*=P_0\times_M\Lambda^{m-1}(\Tan^*M)  \ar[rrd]_{\upsilon}\ar[lld]^{\varrho_0}\ar[ddd]_{\widetilde\tau_0}\ &  \   &
\\
&P_0 \ar[ddrr]^{\bar{\kappa}_0^1}\ar[d]^{\kappa_0^1}\ & \ & \ & \ &\ar[ddll]_{\tau_o}\Lambda^{m-1}(\Tan^*M) 
\\
&E\ar[drr]^{\pi}\ & \ & \ & \ &  
 \\
&\ & \ &M \ & \ & 
}
$$
Finally, we construct the form
$$
\Theta^0_{\cal H}=-\varrho_0^{\,*}\Theta_{\bf h}^0+\d\overline S \in\df^m({\cal P}_0^*)\,,
$$
whose local expression is
$$
\Theta_{\cal H}^0=
\jmath_0^{\,*}(-p_i^\nu\d y^i\wedge\d^{m-1}x_\mu)+H_0\,\d^mx+\d\overline S \,.
$$
Note that, since $\Theta_{\L_s}=FL_{s0}^{\ *}\,\Theta_{\bf h}^0$, we have that $\Theta_\L=\mathcal{F}L_0^{\ *}\,\Theta_{\cal H}^0$.

\begin{proposition}
The pair $(\Theta_{\cal H}^0,\omega)$ is a premulticontact 
%(resp. multicontact) 
structure if, and only if,
$(\Theta_\L,\omega)$ is a premulticontact structure.
\end{proposition}
\begin{proof}
Since ${\cal F}L_0$ is a submersion, for every $Z\in\vf({\cal P}_0^*)$ there exist
$Y\in\vf(J^1\pi)$ such that ${\cal F}L_{0*}Y=Z$.
Therefore, as $\Theta_\L={\cal F}L_0^{\ *}\,\Theta_{\cal H}^0$,
taking into account that ${\cal F}L_0$ is a submersion, we have that 
$$
0=\inn(Y)\Theta_\L=\inn(Y)({\cal F}L_0^{\ *}\,\Theta_{\cal H}^0)={\cal F}L_0^{\ *}[\inn(Z)\Theta_{\cal H}^0]
\ \Longleftrightarrow \ \inn(Z)\Theta_{\cal H}^0=0 \,,
$$
and the same holds for $\d\Theta_\L={\cal F}L_0^{\ *}\,\d\Theta_{\cal H}^0$ and
$\omega={\cal F}L_0^{\ *}\,\omega$.
Hence
$Z\in\Gamma(\ker\omega\cap\ker\Theta_{\cal H}^0\cap\ker\d\Theta_{\cal H}^0)\equiv\Gamma({\cal C}_{\cal H})$ 
if, and only if, $Y\in\Gamma(\ker\omega\cap\ker\Theta_\L\cap\ker\d\Theta_\L)\equiv\Gamma({\cal C}_L)$.
Note that 
$\rk{\cal C}_\L=\rk{\cal C}_{\cal H}+\rk(\ker\,{\cal F}\L_{0*})$

In the same way we can prove that, if $R_{\cal H}\in\vf({\cal P}_0^*)$ and 
$R_\L\in\vf(J^1\pi)$ is such that ${\cal F}L_{0*}R_\L=R_{\cal H}$, then
$R_{\cal H}\in\mathfrak{R}_{\cal H}$ is a Reeb vector field for $(\Theta_{\cal H}^0,\omega)$ if, and only if,
$R_\L\in\mathfrak{R}_\L$ is a Reeb vector field for $(\Theta_\L,\omega)$.
As a consequence  
$\rk{\cal D}^{\mathfrak{R}_\L}=\rk{\cal D}^{\mathfrak{R}_{\cal H}}+\rk(\ker\,{\cal F}\L_{0*})$.
\end{proof}

\begin{definition}
The triple $({\cal P}_0^*,\Theta_{\cal H}^0,\omega)$ 
is the \textbf{premulticontact Hamiltonian system}
associated to the premulticontact Lagrangian system $({\cal P},\Theta_\L,\omega)$.
\end{definition}

\begin{remark}
{\rm
In the premulticontact case,
for the premulticontact Hamiltonian system $({\cal P}_0^*,\Theta_{\cal H}^0,\omega)$, 
the field equations for sections $\psi\colon M\to{\cal P}_0^*$, 
%of $\widetilde\tau_0\colon{\cal P}^*\to M$
multivector fields ${\bf X}_{{\cal H}_0}\in\vf^m({\cal P}_0^*)$,
and Ehresmann connections $\nabla_{{\cal H}_0}$ on ${\cal P}_0^*$
are like \eqref{sect2H}, \eqref{sect1H}, \eqref{vfH}, and \eqref{EcH},
with $\Theta_{\cal H}^0$ instead of $\Theta_{\cal H}$.
In general, these equations are not compatible on ${\cal P}_0^*$ 
and the usual constraint algorithm must be implemented in order to find the final constraint submanifold ${\cal P}_f^*\hookrightarrow{\cal P}_0^*$
(if it exists) where there are integrable
distributions associated with the solutions to the field equations, which are tangent to ${\cal P}_f^*$.
}\end{remark}

%%%%%%%%%%%%%%%%%%%%%%%%%%%%%%%%%%%%%%%%%%%%%%%%%%%%%%%%%%%%%%%%%
\section{Variational formulation}
\label{vaf}
%%%%%%%%%%%%%%%%%%%%%%%%%%%%%%%%%%%%%%%%%%%%%%%%%%%%%%%%%%%%%%%%

There have been several attempts to generalize the Herglotz variational principle to field theories \cite{Geor-2003}. 
In \cite{LPAF2018} a variational principle is presented for Lagrangians with closed action dependence. 
Later, in \cite{GasMas2022}, a principle based on Lagrange multipliers has been used to derive equations of Gravity with dissipation. 
The problem has been studied in more detail in \cite{GLMR-2022}. 
There, it is shown how the variational principle based on Lagrange multipliers requires that the Lagrangian has closed action dependence, namely the dissipation form is closed. 
This condition is also required to develop higher-order Lagrangian field theories. 
The variational principle for the general case is more involved. In this section, we will see how  the multicontact formalism is related to the Herglotz variational principle for fields when the dissipation form is closed (for further details on the Herglotz variational principle for fields theories see \cite{GLMR-2022}).

Consider a (pre)multicontact Lagrangian system $({\cal P},\Theta_\Lag,\omega)$.
Let $\Gamma_{\mathrm{hol}}(\tau)$ be the set of holonomic sections $\bm{\psi}\colon M\to {\cal P}$. We define the following action functional
\begin{equation}
%\label{eqn:DefnFunctionalVariational}
\begin{array}{rcl}
{\mathbb A}_{\cal P} \colon \Gamma_{\mathrm{hol}}(\tau) & \longrightarrow & \Real \\
\bm{\psi} & \longmapsto & \displaystyle \int_M \bm{\psi}^*\d\overline{S}\,,
\end{array}
\label{Lagact}
\end{equation}
where the convergence of the integral is assumed.

\begin{definition}
[Herglotz Principle for Fields: Lagrangian case]
\label{dfn:HerglotzLag}
The {\rm Herglotz variational problem} for the (pre)multicontact Lagrangian system $({\cal P},\Theta_\Lag,\omega)$
consists in finding sections $\bm{\psi}\in\Gamma_{\mathrm{hol}}(\tau)$
satisfying that $\bm{\psi}^*\Theta_\Lag=0$ 
and that are critical sections of the functional ${\mathbb A}_{\cal P}$ with respect to the variations
of $\bm{\psi}$ given by $\bm{\psi}_s = \sigma_s \circ\bm{\psi}$,
where $\left\{ \sigma_s \right\}$ is a local one-parameter group of
any compact-supported $\tau$-vertical vector field $Z$ in ${\cal P}$; that is, it verifies that
\begin{equation*}
%\label{eqn:DynEqVarUnified}
\bm{\psi}^*\Theta_{\cal L}=0 \,, \qquad
\left.\frac{\d}{\d s}\right|_{s=0}\int_M \bm{\psi}_s^*\d\overline{S} = 0 \,.
\end{equation*}
\end{definition}

\begin{proposition}
\label{prop:variational}
Let $({\cal P},\Theta_\Lag,\omega)$ be a closed (pre)multicontact Lagrangian system, namely $\sigma_{\Theta_\Lag}$ is closed. 
A holonomic section $\bm{\psi}\colon M\to {\cal P}$ is a solution to the (pre)multicontact Lagrangian equations
\eqref{ELmcontact} or \eqref{ELmcontact2} if, and only if, it is a solution to the Herglotz Principle \ref{dfn:HerglotzLag}.
\end{proposition}
\begin{proof}
Following the ideas of \cite{GLMR-2022,GasMas2022}, we consider the Herglotz principle as a constraint variational principle. 
Then, $\bm{\psi}$ is a solution to the Herglotz principle if, and only if, there exists a function $\lambda\in \Cinfty(M)$, the Lagrange multiplier, such that $\bm{\psi}$ is critical for the functional
\begin{equation*}
\int_M \bm{\psi}_s^*\left(\d\overline{S}+\lambda \Theta_{\cal{L}} \right) \,.
\end{equation*}
Let $Z$ be a compact-supported $\tau$-vertical vector field on $\cal{P}$,
and $U \subset M$ an open set such that $\partial U$ is a $(m-1)$-dimensional
manifold and $\tau({\rm supp}(Z)) \subset U$. The corresponding variation is given by $\bm{\psi}_s=\sigma_s\circ\bm{\psi}$,
where $\left\{ \sigma_s \right\}$ is a local one-parameter group of
 $Z$. Then, as a consequence of Stoke's theorem and the assumptions made
on the supports of the vertical vector fields, we have that
\begin{align*}
    0 &= \left.\frac{\d}{\d s}\right|_{s=0}\int_M \bm{\psi}_s^*\left(\d\overline{S}+\lambda \Theta_{\cal{L}} \right) = \int_U \bm{\psi}^*\left(\Lie_Z\d \overline{S}+\lambda \Lie_Z\Theta_{\cal{L}} \right)
    \\
    &= \int_U \bm{\psi}^*\left(\d\inn(Z)\d\overline{S}+\lambda \d\inn(Z)\Theta_{\cal{L}} +\lambda \inn(Z)\d\Theta_{\cal{L}}\right)
    \\
    &= \int_U \bm{\psi}^*\left(-\d \lambda\wedge \inn(Z)\Theta_{\cal{L}}+\lambda \inn(Z)\d\Theta_{\cal{L}}\right)+\int_{\partial U} \bm{\psi}^*\left(\inn(Z)\d\overline{S}+ \lambda \inn(Z)\Theta_{\cal{L}}\right)
    \\
    &= \int_U \bm{\psi}^*\left(-\d \lambda\wedge \inn(Z)\Theta_{\cal{L}}+\lambda \inn(Z)\d\Theta_{\cal{L}}\right)\,.
\end{align*}
Thus, we conclude that $\bm{\psi}$ is a solution to the Herglotz variational problem if, and only if,
\begin{equation}\label{eq:variational1}
\bm{\psi}^*\Theta_{\cal{L}}=0\,,\qquad  
\bm{\psi}^*\left(-\d \lambda\wedge \inn(Z)\Theta_{\cal{L}}+\lambda \inn(Z)\d\Theta_{\cal{L}}\right)=0\,,    
\end{equation}
for every compact-supported $\tau$-vertical vector field $Z\in\vf({\cal P})$. 
Since the compact-supported vector fields locally generate the $\Cinfty({\cal P})$-module of vector fields in ${\cal P}$, it follows that the last equality holds for every $\tau$-vertical vector field $Z\in\vf({\cal P})$. 
In particular, if $Z$ is a Reeb vector field $R$, then
$$
0=\bm{\psi}^*\left(-\d \lambda \wedge\inn(R)\Theta_{\cal{L}}+\lambda \inn(R)\d\Theta_{\cal{L}}\right)=
\bm{\psi}^*\left(-\d \lambda \wedge\inn(R)\Theta_{\cal{L}}+\lambda \sigma_{\Theta_{\cal{L}}}\wedge\inn(R)\Theta_{\cal{L}}\right)\,.
$$
Taking $\lambda\neq0$ we have that $\displaystyle \sigma_{\Theta_{\cal{L}}}=\frac{1}{\lambda}\d\lambda$ and,
in particular, $\d\sigma_{\Theta_{\cal{L}}}=0$, as it is required in the hypotheses. Substituting the value of $\lambda$ in \eqref{eq:variational1} we get
\begin{align*}
 0&=\bm{\psi}^*\left(-\d \lambda \wedge\inn(Z)\Theta_{\cal{L}}+\lambda \inn(Z)\d\Theta_{\cal{L}}\right)=
 \lambda\bm{\psi}^*\left(-\sigma_{\Theta_{\cal{L}}} \wedge\inn(Z)\Theta_{\cal{L}}+ \inn(Z)\d\Theta_{\cal{L}}\right)
 \\
 &=\lambda\bm{\psi}^*\left(\inn(Z)(\sigma_{\Theta_{\cal{L}}} \wedge\Theta_{\cal{L}})+\lambda \inn(Z)\d\Theta_{\cal{L}}\right)=  
 \lambda\bm{\psi}^* \inn(Z)\bar{\d}\Theta_{\cal{L}}\,,
\end{align*}
for every $\tau$-vertical vector field $Z\in\vf({\cal P})$.
\end{proof}

In the same way we can state an analogous variational principle for the Hamiltonian case.
In this case, if $({\cal P}^*,\Theta_{\cal H},\omega)$ is a multicontact Hamiltonian system, we have to introduce the functional
\begin{equation*}
%\label{eqn:DefnFunctionalVariational}
\begin{array}{rcl}
{\mathbb A}_{{\cal P}^*} \colon \Gamma(\widetilde\tau) & \longrightarrow & \Real \\
\bm{\psi} & \longmapsto & \displaystyle \int_M \bm{\psi}^*\d\overline{S} \,,
\end{array}
\end{equation*}
where the convergence of the integral is assumed.

\begin{definition}
[Herglotz Principle for Fields: Hamiltonian case]
\label{dfn:HerglotzLag2}
The {\rm Herglotz variational problem} for the multicontact Hamiltonian system $({\cal P}^*,\Theta_{\cal H},\omega)$
consists in finding sections $\bm{\psi}\in\Gamma(\widetilde\tau)$
satisfying that $\bm{\psi}^*\Theta_{\cal H}=0$ 
and that are critical sections of the functional ${\mathbb A}_{{\cal P}^*}$ with respect to the variations
of $\bm{\psi}$ given by $\bm{\psi}_s = \sigma_s \circ\bm{\psi}$,
where $\left\{ \sigma_s \right\}$ is a local one-parameter group of
any compact-supported $\widetilde\tau$-vertical vector field $Z$ in ${\cal P}$; that is, it verifies that
\begin{equation*}
%\label{eqn:DynEqVarUnified}
\bm{\psi}^*\Theta_{\cal H}=0 \,, \qquad
\left.\frac{\d}{\d s}\right|_{s=0}\int_M \bm{\psi}_s^*\d\overline{S} = 0 \,.
\end{equation*}
\end{definition}

Following the same ideas as in Proposition \ref{prop:variational} we can prove the following result.

\begin{proposition}\label{prop:variational1}
Let $({\cal P}^*,\Theta_{\cal H},\omega)$ be a closed (pre)multicontact Hamiltonian system. 
A section $\bm{\psi}\colon M\to {\cal P}^*$ is a solution to the multicontact Hamilton--de Donder--Weyl equations \eqref{sect2H} or \eqref{sect1H} if, and only if, it is a solution to the Herglotz Principle \ref{dfn:HerglotzLag2}.
\end{proposition}

Finally, this result can be generalized to a variational (pre)multicontact system $(P,\Theta,\omega)$ in general. 
Thus, in an open set $U\subset P$ where the adapted coordinates given in Proposition \ref{prop-adapted-coord} are defined, 
we can introduce the form $\overline{S}=s^\mu\,\d^{m-1} x_\mu$ and define the action functional
\begin{equation*}
%\label{eqn:DefnFunctionalVariational}
\begin{array}{rcl}
{\mathbb A}_P \colon \Gamma(\tau) & \longrightarrow & \Real \\
\bm{\psi} & \longmapsto & \displaystyle\int_U \bm{\psi}^*\d\overline{S}=\displaystyle \int_U \bm{\psi}^*(\d s^\mu\wedge\d^{m-1} x_\mu)\,.
\end{array}
\end{equation*}
The corresponding principle is the following.

\begin{definition}
[Herglotz Principle for Fields: General case]
\label{dfn:HerglotzVar}
The {\rm Herglotz variational problem} for the variational (pre)multicontact system $(P,\Theta,\omega)$
consists in finding sections $\bm{\psi}\colon \tau(U)\subset M\to P$
satisfying that $\bm{\psi}^*\Theta=0$ 
and that are critical sections of the functional ${\mathbb A}_P$ with respect to the variations
of $\bm{\psi}$ given by $\bm{\psi}_s = \sigma_s \circ\bm{\psi}$,
where $\left\{ \sigma_s \right\}$ is a local one-parameter group of
any compact-supported $\tau$-vertical vector field $Z$ in $U$.
\end{definition}

Using the same argument as in Proposition \ref{prop:variational}, we have the following result.

\begin{proposition}
Let $(P,\Theta,\omega)$ be a closed variational (pre)multicontact system. A section $\bm{\psi}\colon \tau(U)\subset M\to P$ is a solution to the (pre)multicontact field equations \eqref{sect2} or \eqref{sect1} if, and only if, it is a solution to the Herglotz Principle \ref{dfn:HerglotzVar}.
\end{proposition}

%%%%%%%%%%%%%%%%%%%%%%%%%%%%%%%%%%%%%%%%%%%%%%%%%%%%%%%%%%%%%%%%
\section{Examples}
\label{ex}
%%%%%%%%%%%%%%%%%%%%%%%%%%%%%%%%%%%%%%%%%%%%%%%%%%%%%%%%%%%%%%%%

In this section we give some examples of multicontact field theories. In the first one, we recover the notion of cocontact mechanical system for time-dependent non-conservative systems as a particular case of a multicontact system. 
The second examples describes the Hamiltonian formalism for a vibrating string with time-dependent damping. In the last example, we develop the Lagrangian formalism for the Maxwell's equations with damping. 
In what follows we mainly use the field equations expressed for multivector fields.

%%%%%%%%%%%%%%%%%%%%%%%%%%%%%%%%%%%%%%%%%%%%%%%%%%%%%%%%%%%%%%%%
\subsection{Time-dependent contact mechanical systems}
\label{cocontactsection}

In a recent paper \cite{LGGMR-2022}, a new geometrical framework, called {\sl cocontact structure}, 
has been introduced in order to develop a geometric formulation for time-dependent contact mechanical systems, 
both in the Hamiltonian and Lagrangian settings. 
In this section we see that this cocontact formulation is just a particular case of the multicontact setting introduced in the present paper.

Consider a cocontact Hamiltonian system $(M,\tau,\eta,H)$ where $(M,\tau,\eta)$ is a cocontact manifold of dimension $2n+2$ and $H\colon M\to\R$ is a Hamiltonian function. Recall that on every cocontact manifold there exist local charts $(t, q^i, p_i,s)$ around every point in $M$ such that
$$ \tau = \d t\,,\quad\eta = \d s - p_i\d q^i\,, $$
and the Reeb vector fields read 
$R_t = \tparder{}{t}$ and $R_s = \tparder{}{s}$. 
The cocontact Hamiltonian equations for a vector field $X\in\X(M)$ are
\eqref{hamilton-cocontact-eqs}
%\label{eq:Ham-eq-cocontact-vectorfields}
and the solution to these equations is called the {\sl cocontact Hamiltonian vector field}, which is denoted $X_H$. Its local expression in Darboux coordinates is
\begin{equation}
\label{eq:cocontact-vector-field}
X_H = \parder{}{t} + \parder{H}{p_i}\parder{}{q^i} - \left(\parder{H}{q^i} + p_i\parder{H}{s}\right)\parder{}{p_i} + \left(p_i\parder{H}{p_i} - H\right)\parder{}{s}\,.
\end{equation}

Every cocontact structure $(\tau,\eta)$ on a manifold $M$ along with a Hamiltonian function $H$ allow us to define a multicontact $1$-form $\Theta\in\df^1(M)$ given by
$$ \Theta = H\tau + \eta\,. $$
In fact, we can take the $1$-form $\omega = \tau$. In this case, we have
$$ \ker\omega = \ker\tau = \left\langle \parder{}{q^i},\parder{}{p_i},\parder{}{s} \right\rangle\,,\qquad
\mathfrak{R} = \left\langle\dparder{}{s}\right\rangle \,.
$$
The conditions stated in Definition \ref{premulticontactdef} hold obviously for this structure taking into account that $k = 0$, $m = 1$, and $N = 2n+1$, and that $\mathcal{A}^0(\ker\omega) = \Cinfty(M)$.
In addition, it is easy to check that $\sigma_{\Theta} = \dparder{H}{s}\d t$. In coordinates, we have
\begin{gather*}
\d\Theta = \parder{H}{q^i}\d q^i\wedge\d t + \parder{H}{p_i}\d p_i\wedge \d t + \parder{H}{s}\d s\wedge\d t + \d q^i\wedge\d p_i \,,\\
\bd\Theta = \left(\parder{H}{q^i} + p_i\parder{H}{s}\right)\d q^i\wedge\d t + \parder{H}{p_i}\d p_i\wedge \d t + \d q^i\wedge\d p_i \,.
\end{gather*}
Consider now a vector field $X\in\X^1(M)$ with local expression
$$ 
X = f\parder{}{t} + F^i\parder{}{q^i} + G_i\parder{}{p_i} + g\parder{}{s}\,.
$$
Imposing equations \eqref{vf}, we obtain
\begin{align*}
    f &= 1\,,\\
    g &= p_iF^i - H\,,\\
    G_i &= -\left( \parder{H}{q^i} + p_i\parder{H}{s} \right)\,,\\
    F^i &= \parder{H}{p_i}\,,\\
    0 &=\left(\parder{H}{q^i} + p_i\parder{H}{s}\right) F^i + G_i\parder{H}{p_i} \,,
\end{align*}
where the last equation holds identically when the above ones are taken into account.
Thus, the local expression of the vector field $X$ is
$$ X = \parder{}{t} + \parder{H}{p_i}\parder{}{q^i} - \left(\parder{H}{q^i} + p_i\parder{H}{s}\right)\parder{}{p_i} + \left(p_i\parder{H}{p_i} - H\right)\parder{}{s}\,. $$
This local expression coincides with  \eqref{eq:cocontact-vector-field}. 
Thus, we have checked that time-dependent contact mechanics is a particular case of the multicontact setting introduced in the present work.

Conversely, given a multicontact structure $(\Theta,\omega)$ with $m = 1$, consider the $1$-forms $\tau = \omega$ and $\eta = \Theta - H\omega$. Taking into account that $\omega$ is closed, it is easy to check that $\tau\wedge\eta\wedge(\d\eta)^n$ is a volume form on $M$ and thus $(\tau,\eta)$ define a cocontact structure on $M$.

%%%%%%%%%%%%%%%%%%%%%%%%%%%%%%%%%%%%%%%%%%%%%%%%%%%%%%%%%%%%%%%%
\subsection{A vibrating string with time-dependent damping}

In this second example we are going to study the Hamiltonian formulation of a damped vibrating string in the multicontact setting.

It is well known that the vibration of a string can be described by a function $u=u(t,x)$, where $t$ denotes the time, $x$ denotes the position on the string and $u$ is the distance of a point in the string from its equilibrium position. This system can be described by the Hamiltonian function
$$ H = \frac{1}{2\rho}(p^t)^2 - \frac{1}{2\tau}(p^x)^2\,, $$
where $p^t$ and $p^x$ are the momenta with respect to the coordinates $t$ and $x$ respectively. In \cite{GGMRR-2019}
the authors show that we can modify this Hamiltonian function and, using the $k$-contact formalism for field theories with dissipation, model a vibrating string with linear damping. The modified Hamiltonian is
$$ {\cal H} = H - \gamma s^t\,, $$
where $\gamma\in\R$ is the damping constant and $s^t$ is one of the dissipation variables introduced in the $k$-contact formalism.

In this example, we consider a time-dependent damping and hence, $\gamma = \gamma(t)$ is a function of time. 
This situation can be described using the fiber bundle $\tau\colon{\cal P}\to M$, where $M$ is an orientable manifold with coordinates $(t,x)$ 
and volume form $\omega=\d t\wedge\d x$, and ${\cal P}$ has adapted coordinates $(t,x,u,p^t,p^x,s^t,s^x)$. 
Consider the form
$$ 
\Theta_{\cal H} = H\d t\wedge\d x + \eta^t\wedge\d x - \eta^x\wedge\d t\in\df^2({\cal P})\,, 
$$
where $\eta^t = \d s^t - p^t\d u$ and $\eta^x = \d s^x - p^x\d u$. Hence,
\begin{align*}
\Theta_{\cal H} &= H\d t\wedge\d x + \d s^t\wedge\d x - p^t\d u\wedge\d x - \d s^x\wedge\d t + p^x\d u\wedge\d t\,, \\
\d\Theta_{\cal H} &= \d H\wedge\d t\wedge\d x - \d p^t\wedge\d u\wedge\d x + \d p^x\wedge\d u\wedge\d t\\
&= \frac{p^t}{\rho}\d p^t\wedge\d t\wedge\d x - \frac{p^x}{\tau}\d p^x\wedge\d t\wedge\d x + \gamma(t)\d s^t\wedge\d t\wedge\d x - \d p^t\wedge\d u\wedge\d x + \d p^x\wedge\d u\wedge\d t\,.
\end{align*}
The $1$-form $\sigma_{\Theta}$ defined in Proposition \ref{sigma} is 
$\sigma_{\Theta} = \gamma(t)\d t$, and then
\begin{align*}
\bd\Theta_{\cal H} &= \d\Theta_{\cal H} + \sigma_{\Theta}\wedge\Theta_{\cal H} \\
&= \frac{p^t}{\rho}\d p^t\wedge\d t\wedge\d x - \frac{p^x}{\tau}\d p^x\wedge\d t\wedge\d x - \d p^t\wedge\d u\wedge\d x + \d p^x\wedge\d u\wedge\d t - p^t\gamma(t)\d t\wedge\d u\wedge\d x\,.
\end{align*}

Let $\bfX\in\X^2({\cal P})$ be a locally decomposable, integrable and transverse 2-multivector field. Then, $\bfX = X_1\wedge X_2$ where
\begin{align*}
X_1 &= f_1\parder{}{t} + F_1\parder{}{u} + G_1^1\parder{}{p^t} + G_1^2\parder{}{p^x} + g_1^1\parder{}{s^t} + g_1^2\parder{}{s^x}\,,\\
X_2 &= f_2\parder{}{x} + F_2\parder{}{u} + G_2^1\parder{}{p^t} + G_2^2\parder{}{p^x} + g_2^1\parder{}{s^t} + g_2^2\parder{}{s^x}\,.
\end{align*}
Now,
\begin{equation}\label{eq:vibrating-string-cond-0}
0 =\inn(\bfX)\Theta_{\cal H}
=\inn(X_1\wedge X_2)\Theta_{\cal H}=\inn(X_2)\inn(X_1)\Theta
= f_1f_2 H - f_1F_2p^x + f_1g_2^2 - f_2F_1p^t + f_2g_1^1\,.
\end{equation}
On the other hand,
\begin{align*}
    0 =&\, \inn(\bfX)\bd\Theta_{\cal H} 
    =\, \inn(X_1\wedge X_2)\bd\Theta_{\cal H}
    =\, \inn(X_2)\inn(X_1)\bd\Theta_{\cal H}
\\
    =&\, \left(-f_2F_1p^t\gamma(t) - F_1G_2^2 - f_2G_1^1\frac{p^t}{\rho} + f_2G_1^2\frac{p^x}{\tau} + F_2G_1^2\right)\d t\\
    &+ \left(-f_1F^2p^t\gamma(t) - f_1G_2^1\frac{p^t}{\rho} + f_1G_2^2\frac{p^x}{\tau} + F_1G_2^1 - F_2G_1^1\right)\d x \\
    &+ \left(f_1f_2p^t\gamma(t) + f_1G_2^2 + f_2G_1^1\right)\d u \\
    &+ \left(f_1f_2\frac{p^t}{\rho} - f_2F_1\right)\d p^t + \left(-f_1f_2\frac{p^x}{\tau} - f_1F_2\right)\d p^x \,.
\end{align*}
The transversality condition $\inn({\bf X})\omega=1$ implies $f_1=f_2=1$, and hence we get the conditions
\begin{align}
F_1 &= \frac{p^t}{\rho}\,,
\label{eq:vibrating-string-cond-1}\\
F_2 &= -\frac{p^x}{\tau}\,,
\label{eq:vibrating-string-cond-2}\\
G_1^1 + G_2^2 &= -p^t\gamma(t)\,.
\label{eq:vibrating-string-cond-3}
\end{align}
Note that the terms in $\d t$ and $\d x$ also give equation \eqref{eq:vibrating-string-cond-3}.
Now, taking into account conditions \eqref{eq:vibrating-string-cond-1} and \eqref{eq:vibrating-string-cond-2}, equation \eqref{eq:vibrating-string-cond-0} yields
$$
g_1^1 + g_2^2 = \frac{(p^t)^2}{\rho} - \frac{(p^x)^2}{\tau} - H= \frac{(p^t)^2}{\rho} - \frac{(p^x)^2}{\tau} - \gamma(t)s^t\,.
$$
Let $\psi = (t, x, u, p^t, p^x, s^t, s^x)$ be an integral section of $\bfX$. Then, we have that
$$ 
\parder{^2u}{t^2} - \frac{\tau}{\rho}\parder{^2u}{x^2} + \gamma(t)\parder{u}{t} = 0\,, 
$$
which is the equation of a vibrating string with a time-dependent damping.

%%%%%%%%%%%%%%%%%%%%%%%%%%%%%%%%%%%%%%%%%%%%%%%%%%%%%%%%%%%%%%%%
\subsection{Maxwell's equations}

This last example is devoted to study the Lagrangian formulation of Maxwell's equations with charges and currents with a non-conservative term using the multicontact formulation.

A non-conservative version of Maxwell's equations in the context of contact geometry was first derived in \cite{LPAF2018} from variational principles. 
Subsequently, it has been formalized using $k$-contact geometry in \cite{GasMar21,GRR-2022}. 
Here we describe this theory using multicontact geometry. For the physical consequences of the derived equations and its applications to electromagnetism in materials, see \cite{GasMar21,LPAF2018}.

Let $M$ be a $4$-dimensional manifold representing the spacetime, $P \rightarrow M$ the principal bundle with structure group $U(1)$ and $\pi:C\rightarrow M$ the associated bundle of connections (see \cite{CasMu} for more details). Following the multicontact Lagrangian formalism
developed in Section \ref{mlf}, we define $\mathcal{P}=J^1\pi\times_M\Lambda^{m-1}(\Tan^*M) $, 
with local coordinates $(x^\mu, A_\mu,A_{\mu,\nu},s^\mu)$ such that $\omega=\d x^0\wedge \d x^1\wedge \d x^2\wedge\d x^3$. 
The electromagnetic Lagrangian with a linear dissipation term is
\[
L =  -\frac{1}{4\mu_0}g^{\alpha\mu}g^{\beta\nu}F_{\mu\nu}F_{\alpha\beta}-A_\alpha J^\alpha-\gamma_\alpha s^\alpha\,,
\]
where $F_{\mu\nu}=A_{\nu,\mu}-A_{\mu,\nu}$ is the electromagnetic tensor field, $J^\alpha$ and $\gamma_\alpha\in \Cinfty(M)$ are smooth functions (for $0\leq\alpha\leq 3$), $g^{\mu\nu}$ is a metric on $M$ with signature $(+,-,-,-)$ and $\mu_0$ is a constant \cite{GasMar21}. 
In contrast to \cite{GasMar21,LPAF2018}, in the multicontact framework we can assume that $g^{\mu\nu}$, $J^\alpha$ and $\gamma_\alpha$ are generic functions on the spacetime coordinates.

The Lagrangian energy is
\[
 E_\mathcal{L} = A_{\mu,\alpha}\frac{\partial L}{\partial A_{\mu,\alpha}}-L = \frac{1}{\mu_0}g^{\mu\nu}g^{\alpha\beta}A_{\mu,\alpha}F_{\nu\beta}+\frac{1}{4\mu_0}g^{\mu\nu}g^{\alpha\beta}F_{\beta\nu}F_{\alpha\mu}+A_\alpha J^\alpha+\gamma_\alpha s^\alpha\,,
\]
and the Lagrangian form is
\[
\Theta_\mathcal{L} =E_{\mathcal{L}}\wedge \d ^4x+\d s^\mu\wedge\d ^3x_\mu+\frac{1}{\mu_0}g^{\alpha\beta}g^{\mu\nu}F_{\beta\nu}\d A_\alpha\wedge\d ^3x_\mu\,.
\]
Then, $(\mathcal{P},\Theta_\mathcal{L},\omega)$ is a premulticontact Lagrangian system. Indeed,
$$
   \mathcal{C}=\left<\frac{\partial}{\partial A_{\mu,\nu}}+\frac{\partial}{\partial A_{\nu,\mu}}\right>_{\mu,\nu=0,1,2,3}
   \,,\qquad
    \mathcal{D}^\mathfrak{R}=\left<\frac{\partial}{\partial A_{\mu,\nu}}+\frac{\partial}{\partial A_{\nu,\mu}},\frac{\partial}{\partial s^\mu}\right>_{\mu,\nu=0,1,2,3}\,, 
$$
and $\displaystyle \inn\left(\frac{\partial}{\partial s^\mu}\right)\Theta_{\mathcal{L}}=\d^3x_\mu$. The dissipation form is $\sigma_{\Theta_\mathcal{L}}=\gamma_\mu \d x^\mu$, which is closed as long as $\displaystyle\frac{\partial \gamma_\mu}{\partial x^\nu}=\frac{\partial \gamma_\nu}{\partial x^\mu}$. As it is studied in \cite{GasMar21}, 
for some special values of $\gamma_\mu$ it can be shown that the energy of the system is dissipated. 
In general, we can only say that it is not conserved. The Lagrangian and the premulticontact structure are invariant under a change of gauge. 
Nevertheless, the study of the symmetries of multicontact systems goes beyond the scope of this work.

For a multivector field with local expression
$$
{\bf X}_\L= \bigwedge_{\rho=1}^m
\Big(\derpar{}{x^\rho}+X_{\mu,\rho}\frac{\displaystyle\partial}{\displaystyle
\partial A_\mu}+X_{\mu\nu,\rho}\frac{\displaystyle\partial}{\displaystyle
\partial A_{\mu,\nu}}+X_\rho^\nu\,\frac{\partial}{\partial s^\nu}\Big)\,,
$$
 the field equations \eqref{vf} lead to
\begin{equation}
    \begin{dcases}
    \ X^\mu_\mu = L
    \,,\\
    \ X_{\nu,\mu}-X_{\mu,\nu} = A_{\nu,\mu}-A_{\mu,\nu}\,,\\
    \ \mu_0 J^\mu =g^{\nu\sigma}g^{\mu\tau}\left(X_{\tau\sigma,\nu}-X_{\sigma\tau,\nu}+\gamma_\nu X_{\tau,\sigma}-\gamma_\nu X_{\sigma,\tau}\right) \,. 
    \label{ContactEOM}
    \end{dcases}
\end{equation}
We recover only a part of the holonomy, as it is expected in premulticontact Lagrangian systems. 
Imposing that ${\bf X}_\L$ is holonomic and denoting $X_{\mu\nu,\rho}-X_{\nu\mu,\rho}=F_{\nu\mu,\rho}$, the last equation of \eqref{ContactEOM} is 
$$\mu_0 J^\mu =g^{\nu\sigma}g^{\mu\tau}\left(F_{\sigma\tau,\rho}+\gamma_\nu F_{\sigma\tau}\right)\,.$$

In order to recover an expression involving the electric and magnetic fields, we consider the case $\displaystyle g^{\mu\nu} = \frac{1}{\sqrt{1+\chi_m}}\text{diag}\Big((1+\chi_e)(1+\chi_m),-1,-1,-1 \Big)$, where $\chi_e$ and $\chi_m$ are the electric and magnetic susceptibilities respectively. Then, letting $\displaystyle\gamma^\mu = \left(\frac{\gamma}{c},\pmb{\gamma}\right)$ and $\displaystyle J^\mu = \left(c\rho,\textbf{j}\right)$, the last equation of \eqref{ContactEOM} reads in vector notation as
\begin{align*}
\frac{\rho}{\epsilon_0}&=(1+\chi_e)\big(\nabla\cdot\textbf{E}+\pmb{\gamma}\cdot\textbf{E}\big)\,,
    \\
\textbf{j}&=\-(1+\chi_e)\epsilon_0\left(\frac{\partial\textbf{E}}{\partial t}+\gamma\textbf{E}\right)+\frac{1}{1+\chi_m}\left(\nabla\times\frac{\textbf{B}}{\mu_0}+\pmb{\gamma}\times\frac{\textbf{B}}{\mu_0}\right) \,,
\end{align*}
which are the {\sl contact Maxwell's equations}. 
These equations are the so-called ``second pair of Maxwell's equations'' and correspond to the {\sl Gauss law (for electric fields)} and the {\sl Amp\`ere--Maxwell law} for linear materials (or also for vacuum) when $\gamma_\nu = 0$. 
The ``first pair of Maxwell equations''; i.e., the  {\sl Gauss law (for magnetic fields)} and the {\sl Faraday--Henry--Lenz law}
are the same as in the vacuum case, 
since these laws are just stating that the curvature of the connection, whose coefficients are the electromagnetic tensor field $F$, is closed.

%%%%%%%%%%%%%%%%%%%%%%%%%%%%%%%%%%%%%%%%%%%%%%%%%%%%%%%%%%%%%%%%
\section{Conclusions and outlook}
%%%%%%%%%%%%%%%%%%%%%%%%%%%%%%%%%%%%%%%%%%%%%%%%%%%%%%%%%%%%%%%%

In this paper we have introduced and studied (pre)multicontact structures
as a generalization of (pre)cocontact and (pre)multisymplectic structures.

We have stated the properties defining {\sl (pre)multicontact manifolds} in general. 
The bundle structure of these kinds of manifolds has been analyzed, also proving a Darboux-type theorem to find charts of adapted coordinates. 
This is useful, in first place, to study these manifolds from an abstract point of view and, at the same time, they are the setting for describing non-conservative first-order classical field theories.
To do the latter, we need to restrict the type of (pre)multicontact structures involved.
This is achieved by introducing an additional condition that ensures that the {\sc pde}s associated with the structure 
(which are the field equations in the non-conservative field theories) are of variational type.
We have characterized the field equations using different geometrical tools (sections, multivector fields and Ehresmann connections) that are useful in different contexts.
Then, the Lagrangian and the Hamiltonian formalisms of non-conservative field theories have been developed for the regular and the almost-regular situations.
Finally, the extension of the Herglotz variational Principle has been studied in the Lagrangian, Hamiltonian and general cases, in order to derive the (pre)multicontact field equations variationally.

As illustrative examples, from this framework we have recovered the {\sl cocontact formulation} describing non-conservative time-dependent mechanics, 
we have analyzed the $1$-dimensional wave equation with time-dependent damping, and Maxwell's electromagnetism with damping terms.

The present work opens many new lines of research.
Some of the possible topics are:
\begin{itemize}
\item
To do an accurate analysis of the case of singular systems and the corresponding Lagrangian and Hamiltonian constraint algorithms,
characterizing geometrically the constraint submanifolds 
and studying the equivalence between the solutions to the Lagrangian and the Hamiltonian field equations.
\item 
To develop a Hamilton--Jacobi equation for non-conservative field theories similar to that for contact Hamiltonian systems \cite{LLLR-2022,LLM-2021} and multisymplectic field theories \cite{LMM-2009}.
\item
To study the reduction of multicontact systems in the presence of a group of symmetries, introducing a good definition of momentum map.
For multisymplectic reduction see, for instance, \cite{Bl-2021,EMR-2018} and references therein.
\item
To study the different types of submanifolds of a multicontact manifold and to find an analogue to the coisotropic reduction.
\item
To study the formalism in Cauchy data space in the multicontact setting.
\item
To establish a bracket similar to that of multisymplectic theories, which in this case will not satisfy Leibniz's rule-like properties, as in the case of contact Hamiltonian systems.
Related to this, it would be interesting to study the relationship
between (pre)multicontact structures and Jacobi bundles.
\item
To provide a geometric framework for new gravitational theories to describe astrophysical observations \cite{GasMas2022,Lazo2022}.
\item
It is expected that the multicontact framework can be used to formulate continuum mechanics equations \cite{AH-94}.
\end{itemize}

%%%%%%%%%%%%%%%%%%%%%%%%%%%%%%%%%%%%%%%%%%%%%%%%%%
\section*{Acknowledgments}
%%%%%%%%%%%%%%%%%%%%%%%%%%%%%%%%%%%%%%%%%%%%%%%%%%

We acknowledge financial support of the 
{\sl Ministerio de Ciencia, Innovaci\'on y Universidades} (Spain), projects PGC2018-098265-B-C33,
PID2019-106715GB-C21, and PID2021-125515NB-C21, 
the ICMAT Severo Ochoa project CEX2019-000904-S, and grant EIN2020-112197, funded by AEI/10.13039/501100011033 and European Union NextGenerationEU/PRTR,
and the financial support for research groups AGRUPS-2022 of the Universitat Polit\`ecnica de Catalunya (UPC).
X.~Rivas acknowledges financial support of the Novee Idee 2B-POB II, project PSP: 501-D111-20-2004310 funded by the ``Inicjatywa Doskonałości-Uczelnia Badawcza'' (IDUB) program.
We thank the referees for their helpful comments.

%%%%%%%%%%%%%%%%%%%%%%%%%%%%%%%%%%%%%%%%%%%%%%%%%%%%%%%%%%%%%%%%
\bibliographystyle{abbrv}
\addcontentsline{toc}{section}{References}
%%%%%%%%%%%%%%%%%%%%%%%%%%%%%%%%%%%%%%%

%\bibliographystyle{abbrv}
%\bibliographystyle{apalike}
%\bibliography{bibliografia.bib}

\end{document}